\pdfoutput=1
\RequirePackage{ifpdf}
\ifpdf 
\documentclass[pdftex]{sigma}
\else
\documentclass{sigma}
\fi

\numberwithin{equation}{section}

\newtheorem{Theorem}{Theorem}[section]
\newtheorem*{Theorem*}{Theorem}
\newtheorem{Corollary}[Theorem]{Corollary}
\newtheorem{Lemma}[Theorem]{Lemma}
\newtheorem{Proposition}[Theorem]{Proposition}
 { \theoremstyle{definition}
\newtheorem{Definition}[Theorem]{Definition}

\newtheorem{Example}[Theorem]{Example}
\newtheorem{problem}[Theorem]{Riemann--Hilbert Problem}
\newtheorem{Remark}[Theorem]{Remark} }

\usepackage{mathtools}
\usepackage{tikz}
\usetikzlibrary{arrows}
\usetikzlibrary{decorations.pathmorphing}
\usetikzlibrary{decorations.markings}
\usetikzlibrary{patterns}
\usetikzlibrary{automata}
\usetikzlibrary{positioning}
\usepackage{tikz-cd}
\tikzset{->-/.style={decoration={
			markings,
			mark=at position #1 with {\arrow{latex}}},postaction={decorate}}}

\tikzset{-<-/.style={decoration={
			markings,
			mark=at position #1 with {\arrowreversed{latex}}},postaction={decorate}}}

\usetikzlibrary{shapes.misc}\tikzset{cross/.style={cross out, draw,
		minimum size=2*(#1-\pgflinewidth),
		inner sep=0pt, outer sep=0pt}}

\DeclareMathOperator{\Tr}{Tr}
\def\e{\mathrm{e}}
\newcommand{\R}{\mathbb{R}}
\newcommand{\C}{\mathbb{C}}
\newcommand{\N}{\mathbb{N}}
\newcommand{\Ai}{{\rm Ai}}
\def\d{\mathrm{d}}
\def\i{\mathrm{i}}

\begin{document}
\allowdisplaybreaks

\renewcommand{\thefootnote}{}

\newcommand{\arXivNumber}{2211.16898}

\renewcommand{\PaperNumber}{030}

\FirstPageHeading

\ShortArticleName{Recursion Relation for Toeplitz Determinants and the Discrete Painlev\'e~II Hierarchy}

\ArticleName{Recursion Relation for Toeplitz Determinants\\ and the Discrete Painlev\'e~II Hierarchy\footnote{This paper is a~contribution to the Special Issue on Evolution Equations, Exactly Solvable Models and Random Matrices in honor of Alexander Its' 70th birthday. The~full collection is available at \href{https://www.emis.de/journals/SIGMA/Its.html}{https://www.emis.de/journals/SIGMA/Its.html}}}

\Author{Thomas CHOUTEAU~$^{\rm a}$ and Sofia TARRICONE~$^{\rm b}$}

\AuthorNameForHeading{T.~Chouteau and S.~Tarricone}

\Address{$^{\rm a)}$~Universit\'e d'Angers, CNRS, LAREMA, SFR MATHSTIC, F-49000 Angers, France}
\EmailD{\href{mailto:thomas.chouteau@univ-angers.fr}{thomas.chouteau@univ-angers.fr}}

\Address{$^{\rm b)}$~Institut de Physique Th\'eorique, Universit\'e Paris-Saclay, CEA, CNRS,\\
\hphantom{$^{\rm b)}$}~F-91191 Gif-sur-Yvette, France}
\EmailD{\href{mailto:sofia.tarricone@ipht.fr}{sofia.tarricone@ipht.fr}}
\URLaddressD{\url{https://starricone.netlify.app/}}

\ArticleDates{Received December 22, 2022, in final form May 16, 2023; Published online May 28, 2023}

\Abstract{Solutions of the discrete Painlev\'e~II hierarchy are shown to be in relation with a~family of Toeplitz determinants describing certain quantities in multicritical random partitions models, for which the limiting behavior has been recently considered in the literature. Our proof is based on the Riemann--Hilbert approach for the orthogonal polynomials on the unit circle related to the Toeplitz determinants of interest. This technique allows us to construct a new Lax pair for the discrete Painlev\'e~II hierarchy that is then mapped to the one introduced by Cresswell and Joshi.}


\Classification{33E17; 33C47; 35Q15}

\renewcommand{\thefootnote}{\arabic{footnote}}
\setcounter{footnote}{0}

	\section{Introduction}

Let us consider the symbol $\varphi(z)=\e^{w(z)}$ with
\begin{equation}\label{eqintro: function w and v}
w(z)\coloneqq v(z)+v\big(z^{-1}\big) \qquad\text{and}\qquad v(z)\coloneqq \sum_{j=1}^N\frac{\theta_j}{j}z^j,
\end{equation}
for $\theta_j$ being real constants and natural $N\geq 1$. The $n$-th Toeplitz matrix associated to this symbol and denoted by $T_n(\varphi)$ is a square $(n+1)$-dimensional matrix which entries are given by
\begin{equation}\label{eq:Tn}
T_n(\varphi)_{i,j} \coloneqq \varphi_{i-j}, \qquad i,j=0,\dots,n.
\end{equation}
Here for every $k\in\mathbb{Z}$, $\varphi_k$ is the $k$-th Fourier coefficient of $\varphi(z)$, namely
\begin{equation*}
\label{eq:phi k intro}
\varphi_k=\int_{-\pi}^{\pi}\e^{-{\rm i}k\beta}\varphi\big(\e^{{\rm i}\beta}\big)\frac{{\rm d}\beta }{2\pi},
\end{equation*}
so that $\sum_{k\in\mathbb{Z}}\varphi_k z^k=\varphi(z)$. Notice that, even though it is not emphasized in our notation, the functions $\varphi_k$ and thus the Toeplitz matrix $T_n(\varphi)$ explicitly depend on the natural parameter $N$ which enters in the definition of $v(z)$ in equation~\eqref{eqintro: function w and v}.

In the present work, it is indeed the dependence on this parameter $N$ that we want to study. In particular, we show that the Toeplitz determinants associated to $T_n(\varphi)$, naturally defined as
 \begin{equation}
 \label{def:Dn}
 D_n^N \coloneqq D_n= \det (T_n(\varphi)),
 \end{equation}
are related to some solutions of a discrete version of the Painlev\'e~II hierarchy, indexed over the parameter $N$ (the dependence on $N$ is dropped in the rest of the paper). Our interest in these Toeplitz determinants comes from their appearance in the recent paper~\cite{betea2021multicritical}. The authors there consider some probability measures on the set of integer partitions called \textit{multicritical} Schur measures, which are a particular case of Schur measures introduced by Okounkov in~\cite{Okunkov}.
They are generalizations of the classical Poissonized Plancherel measure and they are defined as
\begin{equation}\label{def:P}
\mathbb{P}(\lbrace \lambda \rbrace ) = Z^{-1} s_{\lambda}[\theta_1,\dots, \theta_N]^2,\qquad\text{with}\quad
Z=\exp\Bigg(\sum_{i=1}^N\frac{\theta_i^2}{i}\Bigg).
\end{equation}
Here $s_{\lambda}[\theta_1,\dots, \theta_N]$ denotes a Schur symmetric function indexed by a partition $\lambda$ that can be expressed as
\[
s_{\lambda}[\theta_1,\dots, \theta_N]=\det_{i,j}h_{\lambda_i-i+j}[\theta_1,\dots,\theta_N],
\]
where $\sum_{k\geq 0}h_k z^k = \exp\big(\sum_{i=1}^N\frac{\theta_i}{i}z^i\big)$. In~\cite{betea2021multicritical}, the authors first used the term \textit{multicritical} to underline that they obtained a different limiting edge behavior for these Schur measures compared to the classical case of the Poissonized Plancherel measure ($N=1$) which is characterised by the Tracy--Widom GUE distribution.
	For more details, we remind to their Theorem~1 or our discussion in the paragraph ``\textit{Continuous limit}'' below, for instance see equation~\eqref{eq:limpn=FN} where the higher order Tracy--Widom distributions appear.

In this setting, denoting by $\lambda=(\lambda_1\geq\lambda_2\geq\dots\geq 0)$ a generic integer partition and by $\lambda'=(\lambda_1'\geq \lambda_2'\geq\dots \geq 0)$ its conjugate partition (namely such that $\lambda_j'=|i:\lambda_i\geq j|$), major quantities of interest of the model are, for any given $n \in \mathbb{N}$,
\begin{equation}
r_n \coloneqq \mathbb{P}(\lambda_1 \leq n) \qquad \text{and}\qquad
q_n\coloneqq \mathbb{P}(\lambda_1' \leq n),
\label{def:pnqn}
\end{equation}
that are often called discrete gap probabilities as random partitions have a natural interpretation in terms of random configuration of points on the set of semi-integers. Indeed, associating the set $\{\lambda_i-i+1/2\}\subset \mathbb{Z}+\frac{1}{2}$ to a partition $\lambda$ (see~\cite{Okunkov}), $r_n$ and $q_n$ can be expressed in terms of a Fredholm determinant of a discrete kernel which corresponds to the gap probability in the determinantal point process defined through the same kernel.

According to Geronimo--Case/Borodin--Okounkov formula~\cite{BO99}, there is a relation between this Fredholm determinant and the Toeplitz determinant $D_n$ and this implies that $r_n$ and $q_n$ (up to a constant factor) are Toeplitz determinants. It leads to (for instance~\cite[Propositions~6 and 7]{betea2021multicritical}):
\begin{gather}
\label{eq:qn}
q_n = \e^{-\sum_{j=1}^N\theta_j^2/j} D_{n-1}.
\end{gather}
For $r_n $ instead, one should define $\widetilde{\theta}_i=(-1)^{i-1}\theta_i$ and by taking $\tilde{w}(z)= \tilde{v}(z)+ \tilde{v}\big(z^{-1}\big)$, where $\tilde{v}(z)$ is nothing than $v(z)$ with $\theta_i$ replaced by $\tilde{\theta}_i$ as given above, the Toeplitz determinant $\widetilde{D}_n$ associated to the symbol $ \widetilde{\varphi}(z)=\e^{\tilde{w}(z)}$ would give the analogue formula
\begin{equation*}
	\label{eq:pn}
	r_n = \e^{-\sum_{j=1}^N\widetilde{\theta}_j^2/j} \widetilde{D}_{n-1}.
\end{equation*}

Notice that in the simplest case, when $N=1$, the quantities $r_n$ and $q_n $ coincide. Moreover, thanks to Schensted's theorem~\cite{Schensted}, they are also equal to the discrete probability distribution function of the length of the longest increasing subsequence of random permutations of size $m$, with $m$ distributed as a Poisson random variable.

 In the case $N=1$, the relation of these quantities with the theory of discrete Painlev\'e equations was shown two decades ago independently and through very different methods by Borodin~\cite{Borodindiscrete}, Baik~\cite{BaikRHforlastpassageperc}, Adler and van~Moerbeke~\cite{AV} and Forrester and Witte~\cite{FW04}.\footnote{They obtained an analogue of equation~\eqref{eq:N=1intro} for Toeplitz determinant associated to symbols which are not necessarily positive or even real valued.} In particular, they all proved that for every $n \geq 1$, the following chain of equalities holds
\begin{equation}
\label{eq:N=1intro}
\frac{D_{n}D_{n-2}}{D_{n-1}^2}=\frac{q_{n+1}q_{n-1}}{q_{n}^2}=\frac{r_{n+1}r_{n-1}}{r_{n}^2}=1-x_n^2,
\end{equation}
where $x_n$ solves the second order nonlinear difference equation
\begin{equation}
\label{eq:dPIIintro}
\theta_1(x_{n+1}+x_{n-1})\big(1-x_n^2\big)+n x_n=0,
\end{equation}
with certain initial conditions.
Equation~\eqref{eq:dPIIintro} is a particular case of the so called discrete Painlev\'e~II equation~\cite{GrammaticosdiscrPain}, a \textit{discrete analogue} of the classical second order ODE known as the Painlev\'e~II equation~\cite{painlev00}. This means that performing some continuous limit of equation~\eqref{eq:dPIIintro} one gets back the Painlev\'e~II equation. The Painlev\'e~II equations, discrete and continuous ones, depend in general on an additional constant term $\alpha\in\R$. In the present work, we consider the discrete Painlev\'e~II equation and its hierarchy in the homogeneous case where $\alpha=0$. Its continuous limit will correspond as well to the case $\alpha=0$.

\begin{Remark}
\label{rk:asympxn}
The homogeneous Painlev\'e~II equation admits a famous solution~\cite{HastingsMcLeod}, called the Hastings--McLeod solution, found by requiring a specific boundary condition at $\infty$. In parallel, one might wonder what is the large $n$ behavior of the solution $x_n$ of the discrete Painlev\'e~II equation~\eqref{eq:dPIIintro}. Its behavior is expressed in terms of the Bessel functions. First, this is suggested by the following heuristic arguments. Because of the definition of $r_n$~\eqref{def:pnqn}, as $n\to\infty$, $r_n$ tends to one and according to the equation~\eqref{eq:N=1intro}, $x_n$ tends to zero. Then for large $n$, the nonlinear term in equation~\eqref{eq:dPIIintro} is small compared to the linear ones and the equation~\eqref{eq:dPIIintro} reduces to the equation
\begin{equation*}
\theta_1\left(x_{n+1}+x_{n-1}\right)+nx_n=0,
\end{equation*}
which indeed admits $J_{-n}(2\theta_1)$, the Bessel function of the first kind of order $-n$, as a solution.
The claim is confirmed by a result of the recent work~\cite{CafassoRuzza}.
The authors there studied the finite temperature deformation for the discrete Bessel point process. The Fredholm determinant of the finite temperature discrete Bessel kernel they studied depends on a function $\sigma$. In the case when $\sigma=1_{\mathbb{Z}_+'}$ (the characteristic function of the set of positive half integers), the Fredholm determinant is then equal to $r_n$. Then from~\cite[equations~(1.33) and~(1.36) of Theorem~III]{CafassoRuzza} together with equation~\eqref{eq:N=1intro}, one can deduce that for $n$ large $x_n^2\sim J_n(2\theta_1)^2$ and, because of the previous discussion, one can conclude
\begin{equation*}
\label{eq:assympxn}
x_n\sim J_{-n}(2\theta_1)=(-1)^nJ_n(2\theta_1),
\end{equation*}
see also Figure~\ref{fig:asympxn}.
\end{Remark}

\begin{figure}[ht]
\centering
\includegraphics[scale=0.62]{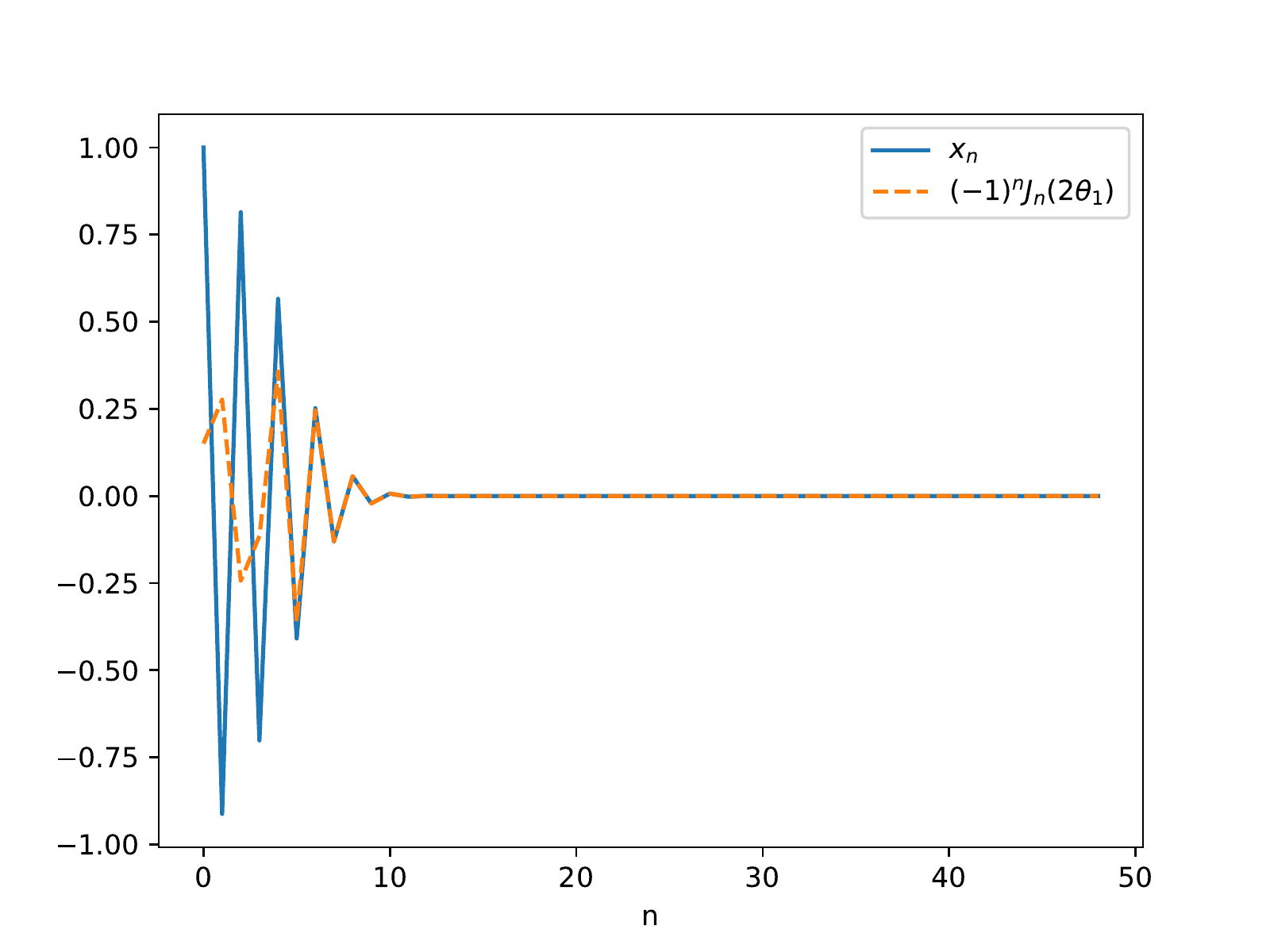}
\caption{For $N=1$, the graphs of $x_n$ and $(-1)^nJ_n(2\theta_1)$ in function of $n$ for $\theta_1=3$.}
\label{fig:asympxn}
\end{figure}

For $N > 1$, Adler and van Moerbeke presented in~\cite{AV}, a generalization of equation~\eqref{eq:N=1intro} by proving that $x_n$ satisfies some recurrence relation written in terms of the Toeplitz lattice Lax matrices. The main result of our work is a recurrence relation for $x_n$ defined via a $N$-times iterating discrete operator which establishes the link with the discrete Painlev\'e~II hierarchy~\cite{joshi1999discrete}.
The precise result is stated as below.
\begin{Theorem}
	\label{thm:main intro}
For any fixed $N\geq 1$, for the Toeplitz determinants $D_n$~\eqref{def:Dn}, $n\geq 1$ associated to the symbol $\varphi(z)$~\eqref{eqintro: function w and v}, we have
\begin{equation}
\label{eq:Ngenintro}
\frac{D_{n}D_{n-2}}{D_{n-1}^2} = 1-x_n^2,
\end{equation}
where $x_n $ solves the $2N$ order nonlinear difference equation
\begin{equation}
\label{eq:dPIIhierarchy intro}
nx_n+\bigl(-v_n-v_n{\rm Perm}_n+2x_n\Delta^{-1}(x_n-(\Delta+I)x_{n}{\rm Perm}_n)\bigr){L}^N(0)=0,
\end{equation}
where $L$ is a discrete recursion operator defined as
\begin{equation}
\label{eq:recop intro}
L(u_n)\coloneqq \big(x_{n+1}\big(2\Delta^{-1}+I\big)((\Delta+I)x_{n}{\rm Perm}_n-x_n) +v_{n+1}(\Delta+I)-x_nx_{n+1}\bigr)u_n.
\end{equation}
Here $v_n\coloneqq 1-x_n^2$, $\Delta$ denotes the difference operator
\begin{equation*}
\Delta\colon\ u_n\to u_{n+1}-u_n
\end{equation*}
and ${\rm Perm}_n$
is the transformation of the space $\C\big[(x_j)_{j\in[[0,2n]]}\big]$ acting by permuting indices in the following way:
\begin{align}\label{eq:Permdef intro}
\begin{split}
{\rm Perm}_n\colon\quad \C\big[(x_j)_{j\in[[0,2n]]}\big]&\longrightarrow\C\big[(x_j)_{j\in[[0,2n]]}\big],
\\
P\big((x_{n+j})_{-n\leqslant j\leqslant n}\big)&\longmapsto P\big((x_{n-j})_{-n\leqslant j\leqslant n}\big).
\end{split}
\end{align}
\end{Theorem}

\begin{Remark}
According to equation~\eqref{eq:dPIIhierarchy intro} and the definition of the operator $L$~\eqref{eq:recop intro}, we need to perform discrete integrations to compute the $N$-th equation of the discrete Painlev\'e~II hierarchy. It is always possible to accomplish this discrete integration. The operator $\Delta^{-1}$, inverse of the difference operator $\Delta$, is applied to $(\Delta+I)x_{n}{\rm Perm}_n-x_n$ and it is possible to write this operator as a derivative. Indeed,
\begin{equation*}
\left(\Delta+I\right)x_{n}{\rm Perm}_n-x_n=\Delta x_n{\rm Perm}_n+({\rm Perm}_n-I)x_n.
\end{equation*}
The first term on the right hand side is a derivative and because of the definition of ${\rm Perm}_n$, the second term can be expressed as a derivative.
\end{Remark}
Equation~\eqref{eq:dPIIhierarchy intro}, together with the definition of the recursion operator $L$ in~\eqref{eq:recop intro}, of the quantity $v_n$ and of the transformation ${\rm Perm}_n$ in~\eqref{eq:Permdef intro} is indeed the $N$-th member of the discrete Painlev\'e~II hierarchy.
The first equations of the hierarchy read as
\begin{align}\label{eq:hierar1 intro}
N=1\colon\ &n x_n+\theta_1(x_{n+1}+x_{n-1})\big(1-x_n^2\big)=0,
\\
N=2\colon\
&n x_n+\theta_1\big(1-x_n^2\big)(x_{n+1}+x_{n-1})
+\theta_2\big(1-x_n^2\big)\nonumber
\\
&\qquad\times{}\bigl(x_{n+2}\big(1-x_{n+1}^2\big)+x_{n-2}\big(1-x_{n-1}^2\big)-x_n(x_{n+1}+x_{n-1})^2\bigr)=0,
\label{eq:hierar2 intro}
\\
N=3\colon\	&nx_n+\theta_1\big(1-x_n^2\big)(x_{n+1}+x_{n-1})
+\theta_2\big(1-x_n^2\big)\bigl(x_{n+2}\big(1-x_{n+1}^2\big)\nonumber
\\
&\qquad{}+x_{n-2}\big(1-x_{n-1}^2\big)-x_n(x_{n+1}+x_{n-1})^2\bigr)
+\theta_3\big(1-x_n^2\big)\bigl(x_n^2(x_{n+1}+x_{n-1})^3\nonumber
\\
&\qquad{}+x_{n+3}\big(1-x_{n+2}^2\big)\big(1-x_{n+1}^2\big)+x_{n-3}\big(1-x_{n-2}^2\big)
\big(1-x_{n-1}^2\big)\bigr)\nonumber
\\
&\qquad{}+\theta_3\big(1-x_n^2\big)\bigl(-2x_n(x_{n+1}+x_{n-1})\big(x_{n+2}\big(1-x_{n+1}^2\big)+x_{n-2}
\big(1-x_{n-1}^2\big)\big)\nonumber
\\
&\qquad{}-x_{n-1}x_{n-2}^2\big(1-x_{n-1}^2\big)\bigr)\nonumber
\\
&\qquad{}+\theta_3\big(1-x_n^2\big)\bigl(-x_{n+1}x_{n+2}^2\big(1-x_{n+1}^2\big)-x_{n+1}x_{n-1}(x_{n+1}+x_{n-1})\bigr)=0
\label{eq:hierar3 intro}
\end{align}
with the first one coinciding with the discrete Painlev\'e~II equation~\eqref{eq:dPIIintro}. Computations with the operator~\eqref{eq:recop intro} introduced in Theorem~\ref{thm:main intro} for $N=1$ and $2$ are done in Example~\ref{ex:computationsDiscetePIIN=12}.
\begin{Remark} The same heuristic argument used in Remark~\ref{rk:asympxn} applies also when $N>1$ (since $x_n$ still tends to zero as $n\to\infty$), thus suggesting that the $N$-th equation of the discrete Painlev\'e~II hierarchy reduces to a linear discrete equation for large $n$. For $N=2$ and $3$, the reduced equations are
\begin{align*}
N=2\colon\	&n x_n+\theta_1(x_{n+1}+x_{n-1})+\theta_2(x_{n+2}+x_{n-2})=0,
\\
N=3\colon\	&n x_n+\theta_1(x_{n+1}+x_{n-1})+\theta_2(x_{n+2}+x_{n-2}) +\theta_3(x_{n+3}+x_{n-3})=0.
\end{align*}
Similar recurrence relations appeared
in~\cite{Dattoli92} for the multivariable generalized Bessel functions (GBFs). These generalized Bessel functions were discussed in~\cite{KimuraZahabi, Okunkov} in the context of Schur measures for random partitions and generalizations of the previous recurrence equations were introduced (in particular, see in~\cite[equation~(3.2b)]{KimuraZahabi}). We denote by $J_n^{(N)}(\xi_1,\dots,\xi_N)$ a $N$-variable GBFs of order $n$. In~\cite{Dattoli92}, $J_n^{(N)}(\xi_1,\dots,\xi_N)$ is defined as a discrete convolution product of $N$ Bessel functions. In particular, if $j_n^{(k)}(\xi)$ is the $n$-th Fourier coefficient of the function $\beta\to\e^{2{\rm i}\xi\sin(k\beta)}$ then
\begin{equation*}
J_n^{(N)}(\xi_1,\dots,\xi_N)\coloneqq j_n^{(N)}(\xi_N)\ast j_n^{(N-1)}(\xi_{N-1})\ast \cdots\ast j_n^{(1)}(\xi_1)(n),
\end{equation*}
where $\ast$ denotes the discrete convolution.

In the case $N=1$, the symbol we considered was $\varphi\big(\e^{{\rm i}\beta}\big)=\e^{\theta_1(\e^{{\rm i}\beta}+\e^{-{\rm i}\beta})}=\e^{2\theta_1\cos(\beta)}$ and the large $n$ asymptotic behavior of $x_n$ was given by $J_{-n}(2\theta_1)$ which is the $n$-th Fourier coefficient of the function $\beta\to \e^{\theta_1(\e^{{\rm i}\beta}-\e^{-{\rm i}\beta})}$ up to a constant $(-1)^n$.

For $N>1$, the symbol is $\varphi_N(\e^{{\rm i}\beta})=\prod_{k=1}^N\e^{\frac{\theta_k}{k}(\e^{{\rm i}k\beta}+\e^{-{\rm i}k\beta})}=\prod_{k=1}^N\e^{2\frac{\theta_k}{k}\cos(k\beta)}$. Then, we conjecture that the large $n$ asymptotic behavior of $x_n^{(N)}$ would be given by the $n$-th Fourier coefficient of $\beta\to \prod_{k=1}^N\e^{\frac{(-1)^{k+1}\theta_k}{k}(\e^{{\rm i}k\beta}-\e^{-{\rm i}k\beta})}$ which is precisely $J_n^{(N)}(\xi_1,\dots,\xi_N)$ up to some constant and proper rescaling:
\begin{equation*}
x_n^{(N)}\sim(-1)^nJ_n^{(N)}\bigg(\bigg((-1)^i\frac{2}{i}\theta_i\bigg)_{1\leqslant i\leqslant N}\bigg),
\end{equation*}
see also Figure~\ref{fig:asympxn2et3}.
\end{Remark}

\begin{figure}[ht]
 \begin{center}
\includegraphics[scale=0.5]{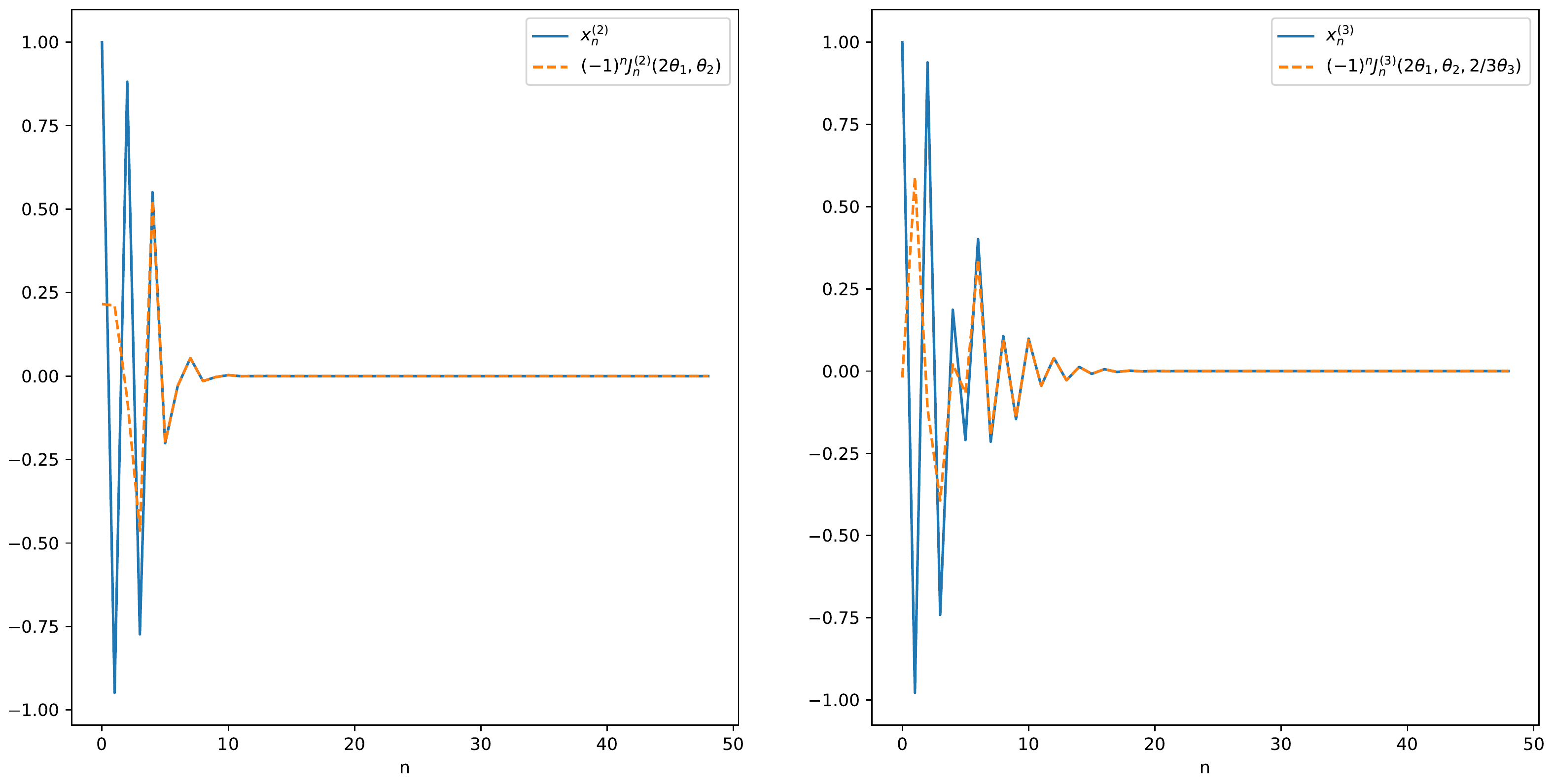}
\caption{The graphs of $x_n^{(N)}$ and $(-1)^nJ_n^{(N)}\big((\theta_i)_{1\leqslant i\leqslant N}\big)$ (for $N=2$ on left and $N=3$ on the right) in function of $n$ for $\theta_1=3$, $\theta_2=1.2$ and $\theta_3=2.6$.}
 \label{fig:asympxn2et3}
 \end{center}
\end{figure}
\begin{Remark}
	Notice that for $N=1,2$ the equations~\eqref{eq:hierar1 intro} and~\eqref{eq:hierar2 intro} coincide with the ones found in~\cite{AV}. Also notice that in the physics literature, Periwal and Schewitz~\cite{periwalschewitz} found similar discrete equations for $N=1,2$ (with different coefficients sign) in the context of unitary matrix models and used their solutions to evaluate the behavior of some typical integrals in the large-dimensional limit passing through the continuous limit of their discrete equations.
For $N=1$, the discrete Painlev\'e~II equation was also found in~\cite{HisakadoUnitary1} as a particular case of the string equation for the full unitary matrix model, i.e., for $w(z)=\theta_1z+\theta_{-1}z^{-1}$. The dependence in $\theta_{\pm1}$ of $x_n$ (and some other $x_n^*$) was also studied there and it produced some evolution equations related, after some change of variables, to the two-dimensional Toda equations. This would suggest that for the general case $N>1$, the dependence of $x_n$ on times $\theta_1,\dots,\theta_N$ would be related to the one-dimensional Toda hierarchy (see also~\cite{Okunkov}).
\end{Remark}

The first construction of a discrete Painlev\'e~II hierarchy in~\cite{joshi1999discrete} used the integrability property of the continuous one, in the following sense. It is well known that the classical Painlev\'e~II equation admits an entire hierarchy of higher order analogues. Indeed, this equation can be obtained as a self-similarity reduction of the modified KdV equation. Thus, the higher order members of the Painlev\'e~II hierarchy are but analogue self-similarity reductions of the corresponding higher order members of the modified KdV hierarchy (see, e.g.,~\cite{flaschkanewell1}). In some way, this implies that the Lax representation of the KdV hierarchy in terms of isospectral deformations becomes for the Painlev\'e~II hierarchy a Lax representation in terms of isomonodromic deformations~\cite{clarkson2006lax}.

In~\cite{joshi1999discrete} then, the discrete Painlev\'e~II hierarchy is defined as the compatibility condition of a~sort of ``discretization'' of the Lax representation of the Painlev\'e~II hierarchy. In particular, they considered the compatibility condition of a linear $2\times 2$ matrix-valued system of the following type:
\begin{equation}\label{eq:lin sys intro}
\Phi_{n+1}(z) = L_n(z) \Phi_n(z), \qquad
\frac{\partial}{\partial z}\Phi_{n}(z) = M_n(z) \Phi_n(z),
\end{equation}
where the coefficients $L_n(z)$, $M_n(z)$ are explicit matrix-valued rational function in $z$, depending on $x_\ell,\ell=n+N,\dots, n-N$, in some recursive (on $N$) way. This allows the authors there to compactly write the $N$-th discrete Painlev\'e~II equation using some recursion operators.
The linear system that we obtain in Proposition~\ref{Prop:LaxPair} and that encodes our hierarchy as written in~\eqref{eq:dPIIhierarchy intro} is mapped into the one of~\cite{joshi1999discrete} through an explicit transformation, as shown in Proposition~\ref{Prop:connectionJoshi}, thus implying that~\eqref{eq:dPIIhierarchy intro} is indeed the same discrete Painlev\'e~II hierarchy.

{\bf Continuous limit.}
The aim of this paragraph is to explain heuristically the reason why our result given in Theorem~\ref{thm:main intro} can be considered as the discrete analogue of the generalized Tracy--Widom formula for higher order Airy kernels (namely, the result contained in~\cite[Theorem~1.1]{CafassoClaeysGirotti}, case $\tau_i=0$).

For $N=1$, Borodin in~\cite{Borodindiscrete} already pointed out that formula~\eqref{eq:N=1intro} with~\eqref{eq:dPIIintro} can be seen as a discrete analogue of the classical Tracy--Widom formula for the GUE Tracy--Widom distribution~\cite{TW94,TracyWidom}.
In other words, he described how to pass from the left to the right in the picture below:
\begin{center}
\begin{tikzpicture}
\node[] at (-9,2) {Discrete case};
\node[] at (-9,1) {$\dfrac{D_nD_{n-2}-D_{n-1}^2}{D_{n-1}^2}=-x_n^2$};
\node[] at (-9,0) {with $nx_n+\theta\big(1-x_n^2\big)(x_{n+1}+x_{n-1})=0$,};
\node[] at (0,2) {Continuous case};
\node[] at (0,1) {$\dfrac{{\rm d}^2}{{\rm d} t^2}\log \det (1-\mathcal{K}_\Ai|_{(t,+\infty)})=-u^2(t)$};
\node[] at (0,0) {with\quad $u''(t)= 2u^3(t)+tu(t)$,};
\node[] at (0,-.8) {$u(t)\underset{t\rightarrow \infty}{\sim} \Ai(t)$,};
\draw[->] (-6,1.8)--(-3.5,1.8);	
\node[] at (-4.75,2.1) {$\text{Baik--Deift--Johansson}$};
\end{tikzpicture}
\end{center}
where $\Ai(t)$ denotes the classical Airy function and $\mathcal{K}_\Ai$ denotes the integral operator acting on $L^2(\R)$ through the Airy kernel.
This connection was achieved by using the scaling limit computed by Baik, Deift and Johansonn in~\cite{baik1999distribution} for the distribution of the first part of partitions in the Poissonized Plancherel random partition model (which is recovered in~\cite[Theorem~1]{betea2021multicritical} for $N=1$). In some way, as emphasized by Borodin, their result not only assures the existence of a limiting function for the $D_n$, in this case $ D(t)= \det (1-\mathcal{K}_\Ai|_{(t,+\infty)})$, for a certain continuous variable $t$. It also encodes already how the discrete function $x_n$, should be rescaled in terms of a differentiable function $u(t)$ to get back, from the recursion relation for $D_n$, the Tracy--Widom formula.

To generalize this result for the case $N>1$, we proceed by adapting the method used by Borodin in~\cite{Borodindiscrete} for $N=1$ to the higher order cases, using the scaling proposed in~\cite{betea2021multicritical}\footnote{Up to the correction of the typo $d\rightarrow d^{-1}$ in their statement of Theorem 1.} for the multicritical case (notice that their $n$ corresponds to our $N$), instead of the Baik--Deift--Johansson's one that only holds for $N=1$.

We recall that $D_n$ is the Toeplitz determinant associated to the symbol $\varphi(z)$~\eqref{eqintro: function w and v} (which depends on $\theta_i$, $i=1,\dots,N$ and thus on $N$). In the following discussion, we write explicitly the dependence on the family of parameters $(\theta_i)$, $i=1,\dots,N$ of $D_n=D_n(\theta_i)$, $x_n=x_n(\theta_i)$, ${r_n=r_n(\theta_i)}$ and $q_n=q_n(\theta_i)$. Consider equation~\eqref{eq:Ngenintro} written in terms of the Toeplitz determinants $D_n(\theta_i)$ in this way
\begin{equation}
	\label{eq:toepl discrete der intro}
\frac{D_{n-2}(\theta_i)D_n(\theta_i)-D_{n-1}^2(\theta_i)}{D_{n-1}^2(\theta_i)}=-x_n^2(\theta_i).
\end{equation}
From the equation~\eqref{eq:qn}, this previous equation can be expressed in terms of $q_n(\theta_i)$ defined as~\eqref{def:pnqn}. It becomes
\begin{equation}
	\label{eq:discrete der qn}
\frac{q_{n-1}(\theta_i)q_{n+1}(\theta_i)-q_{n}^2(\theta_i)}{q_{n}^2(\theta_i)}=-x_n^2(\theta_i).
\end{equation}
 According to~\cite[Lemma~8]{betea2021multicritical}, with the change of parameters $\tilde{\theta}_i=(-1)^{i-1}\theta_i$, we have $q_n(\theta_i)=r_n\big(\tilde{\theta}_i\big)$. Thus equation~\eqref{eq:discrete der qn} now reads as
\begin{equation}
\label{eq:discrete der pn}
\frac{r_{n-1}\big(\tilde{\theta}_i\big)r_{n+1}\big(\tilde{\theta}_i\big)-r_{n}^2\big(\tilde{\theta}_i\big)}{r_{n}^2\big(\tilde{\theta}_i\big)}=-x_n^2(\theta_i).
\end{equation}
Following the scaling limit described in~\cite[Theorem~1]{betea2021multicritical}, we define the following scaling for the discrete variable $n$:
\begin{equation}
	\label{eq: continuous variable}
n=b\theta +t\theta^{\frac{1}{2N+1}}d^{-\frac{1}{2N+1}} \quad \iff \quad
t = (n-b\theta)\theta^{-\frac{1}{2N+1}}d^{\frac{1}{2N+1}}
\end{equation}
with $b$, $d$ defined as
\begin{equation*}
b=\frac{N+1}{N}, \qquad
d=\begin{pmatrix}2N\\N-1\end{pmatrix}
\end{equation*}
and choose $\tilde{\theta}_i$ (resp.\ $\theta_i$) all proportional to $\theta = \tilde{\theta}_1=\theta_1$ in the following way:
\begin{equation*}
\tilde{\theta}_i=(-1)^{i-1}\frac{(N-1)!(N+1)!}{(N-i)!(N+i)!}\theta, \qquad
i=1,\dots, N,
\end{equation*}
respectively,
\begin{equation}
\label{eq:scaling of theta i}
\theta_i=\frac{(N-1)!(N+1)!}{(N-i)!(N+i)!}\theta, \qquad i=1,\dots, N.
\end{equation}
Now recall the definition of $r_n\big(\tilde{\theta}_i\big)$~\eqref{def:pnqn} in function of $\mathbb{P}=\mathbb{P}_{\tilde{\theta}_i}$ (see equation~\eqref{def:P} for the definition of $\mathbb{P}$ and the dependence on the family of parameters $(\theta_i)$). From the previous scaling, it is now possible to express $r_n\big(\tilde{\theta}_i\big)$ in function of $t$ and $\theta$
\begin{equation}
\label{eq:pnPthetat}
r_n\big(\tilde{\theta}_i\big)=\mathbb{P}_{\tilde{\theta}_i}\bigg(\frac{\lambda_1-b\theta}{(\theta d^{-1})^{\frac{1}{2N+1}}}\leqslant t\bigg)
\end{equation}
 and according to~\cite[Theorem 1]{betea2021multicritical}, the limiting behavior of the probability distribution function of $\lambda_1$ in this setting is given by
\begin{gather}
\lim_{\theta \rightarrow +\infty} r_n\big(\tilde{\theta}_i\big)=\lim_{\theta \rightarrow +\infty} \mathbb{P}_{\tilde{\theta}_i} \bigg(\frac{\lambda_1 -b\theta}{\big(\theta d^{-1}\big)^{\frac{1}{2N+1}}}\leqslant t\bigg)=	F_N(t), \nonumber\\
F_N(t)=\det(1-\mathcal{K}_{\Ai_{2N+1}}|_{(t,\infty)}),\label{eq:limpn=FN}
\end{gather}
where $\mathcal{K}_{\Ai_{2N+1}}$ is the integral operator acting with higher order Airy kernel (see, for instance, in~\cite[equation~(2.7)]{betea2021multicritical}).

As we did for $r_n\big(\tilde{\theta}_i\big)$ in equation~\eqref{eq:pnPthetat}, we express $r_{n+1}\big(\tilde{\theta}_i\big)$ and $r_{n-1}\big(\tilde{\theta}_i\big)$ in function of $t$ and~$\theta$:
\begin{equation*}
r_{n\pm 1}\big(\tilde{\theta}_i\big)=\mathbb{P}_{\tilde{\theta}_i}\bigg(\frac{\lambda_1-b\theta}{\big(\theta d^{-1}\big)^{\frac{1}{2N+1}}}\leqslant t\pm\big(\theta d^{-1}\big)^{-\frac{1}{2N+1}}\bigg).
\end{equation*}
With this discussion and this scaling for $n$, $(\theta_i)$ and $\big(\tilde{\theta}_i\big)$, we deduce that{\samepage
\begin{equation*}
-\lim_{\theta \rightarrow +\infty}\dfrac{x_n^2(\theta_i)}{\big(\theta d^{-1}\big)^{-\frac{2}{2N+1}}}=\lim_{\theta \rightarrow +\infty}\frac{r_{n-1}\big(\tilde{\theta}_i\big)r_{n+1}\big(\tilde{\theta}_i\big)-r_{n}^2\big(\tilde{\theta}_i\big)}{\big(\theta d^{-1}\big)^{-\frac{2}{2N+1}}r_{n}^2\big(\tilde{\theta}_i\big)}= \frac{{\rm d}^2}{{\rm d}t^2} \log F_N(t),
\end{equation*}
where the first equality comes from equation~\eqref{eq:discrete der pn} and the second from equation~\eqref{eq:limpn=FN}.}

From now on, we drop the dependence on $\theta_i$, $i=1,\dots,N$ in the notation. The previous equation suggests that, in order to be consistent with~\cite[Theorem~1.1]{CafassoClaeysGirotti}, the discrete function $x_n$ appearing in formula~\eqref{eq:toepl discrete der intro} in the scaling~\eqref{eq: continuous variable} for $n$ and~\eqref{eq:scaling of theta i} for $(\theta_i)$ limit should be
\begin{equation*}
-x_n^2 \sim -(\theta)^{-\frac{2}{2N+1}}d^{\frac{2}{2N+1}} u^2(t)
\end{equation*}
with $u(t)$ solution of the $N$-th equation of the Painlev\'e~II hierarchy. This can be proved directly by computing the scaling limit of the equations of the discrete Painlev\'e~II hierarchy we found for $x_n$ in Theorem~\ref{thm:main intro}. Indeed, for every fixed $N$, we write $x_n$ as
\begin{equation}
	\label{eq: continuous xn scaling}
x_n=(-1)^n \theta^{-\frac{1}{2N+1}}d^{\frac{1}{2N+1}} u(t)
\end{equation}
with $u(t)$ a smooth function of the variable $t$ defined as in equation~\eqref{eq: continuous variable}.
Now recall that $x_n$ solves the discrete equation~\eqref{eq:dPIIhierarchy intro} of order $2N$ for every $N\geq 1$. The continuous limit of the discrete equations of the hierarchy~\eqref{eq:dPIIhierarchy intro}, under the definition of $x_n$~\eqref{eq: continuous xn scaling} and the scaling of the parameters $\theta_i$ as~\eqref{eq:scaling of theta i}, gives the equations of the classical Painlev\'e~II hierarchy. For any fixed~$N$ the computation should be done in the same way: consider the $N$-th discrete equation of the hierarchy~\eqref{eq:dPIIhierarchy intro} and replace each $\theta_i$ with the values given in formula~\eqref{eq:scaling of theta i}. Then substitute~$x_n$ with the definition in~\eqref{eq: continuous xn scaling} and for $\theta\rightarrow+\infty$ compute the asymptotic expansion of every term $x_{n+K}\propto u\big(t+K\theta^{-\frac{1}{2N+1}}d^{\frac{1}{2N+1}}\big)$, $K=-N,\dots,N$ appearing in the discrete equation. The coefficient of $\theta^{-1}$ resulting after this procedure coincides indeed with the $N$-th equation of the Painlev\'e~II hierarchy.
For $N=1,2,3$, the computations are explicitly done in the Appendix~\ref{appendix}.

\begin{Remark}
It is worthy to mention that in~\cite{CafassoClaeysGirotti}, the authors also consider a generalization of the Fredholm determinant $F_N(t)$, recalled here in~\eqref{eq:limpn=FN}, depending on additional parameters~$\tau_i$. Those are related to solutions of the general Painlev\'e~II hierarchy, which depends as well on the~$\tau_i$. With the scaling as in~\cite{betea2021multicritical} for the $\theta_i$'s, the continuous limit for our discrete equations leads to the Painlev\'e~II hierarchy with $\tau_i=0$ for all~$i$. This is consistent with the fact that the limiting behavior in~\cite{betea2021multicritical}, written here in equation~\eqref{eq:limpn=FN}, involves indeed the Fredholm determinant $F_N(t)$ corresponding to $\tau_i=0$ for all $i$ (the same already appeared in~\cite{LDMS18}). \end{Remark}

{\bf Methodology and outline.} The rest of the work is devoted to prove Theorem~\ref{thm:main intro}. In~order to do so, we introduce the classical Riemann--Hilbert characterization~\cite{BaikDeiftSuidan} of the family of orthogonal polynomials on the unit circle (OPUC for brevity) with respect to a measure defined by the symbol $\varphi(z)$. Classical results from orthogonal polynomials theory allow to achieve almost directly formula~\eqref{eq:toepl discrete der intro} where $x_n$ is defined as the constant term of the $n$-th monic orthogonal polynomial of the family. The Riemann--Hilbert problem for the OPUC is then used to deduce a linear system of the same type of~\eqref{eq:lin sys intro} which is proven to be in relation with the Lax pair introduced by Cresswell and Joshi~\cite{joshi1999discrete} for the discrete Painlev\'e~II hierarchy. This is done in Section~\ref{sec2}. The explicit computation of the Lax pair together with the construction of the recursion operator and the hierarchy for $x_n$ as written in~\eqref{eq:dPIIhierarchy intro} are done in Section~\ref{sec3}.

	\section[OPUC: the Riemann--Hilbert approach and a discrete Painlev\'e~II Lax pair]{OPUC: the Riemann--Hilbert approach and a discrete \\Painlev\'e~II Lax pair}
	\label{sec2}
In this section, we introduce the relevant family of orthogonal polynomials on the unit circle. We recall some of their properties and their Riemann--Hilbert characterization. Afterward we derive a Lax pair associated to the Riemann--Hilbert problem and establish the relation with the Lax pair for discrete Painlev\'e~II hierarchy~\eqref{eq:lin sys intro} introduced by Cresswell and Joshi~\cite{joshi1999discrete}. The proofs of the results for orthogonal polynomials stated in here can be found in the classical reference~\cite{BaikDeiftSuidan}.

We denote by	$S^1$ the unit circle in $\C$ counterclockwise oriented.
	We consider the following positive measure on $S^1$ (absolutely continuous w.r.t.\ the Lebesgue measure there):
	\begin{equation}
		\label{eq: measure}
		\d\mu(\beta) = \frac{\e^{w(\e^{{\rm i}\beta})}}{2\pi}\d\beta,
	\end{equation}
	where the function $w(z)$ for any $z\in \mathbb{C}$ is given as in equation~\eqref{eqintro: function w and v}.
	The family of orthogonal polynomials on the unit circle (OPUC) w.r.t.\ the measure~\eqref{eq: measure} is defined as the collection of polynomials $\lbrace p_n(z) \rbrace_{n \in \N}$ written as
	\begin{equation}\label{eq: n th opuc}
		p_n(z)=\kappa_{n}z^n + \dots +\kappa_0, \qquad
\kappa_n >0
	\end{equation}
	and such that the following relation holds for any indices $k$, $h$
	\begin{equation*}
		\int_{-\pi}^{\pi} \overline{p_k\big (\e^{\i\beta}\big)}p_h\big(\e^{\i\beta}\big)\frac{\d \mu(\beta)}{2\pi}=\delta_{k,h}.
	\end{equation*}
	The family of monic orthogonal polynomials $\lbrace\pi_n(z)\rbrace$ associated to the previous ones is defined in analogue way, so that $p_n(z)=\kappa_{n}\pi_n(z)$.
	
	\subsection{Toeplitz determinants related to OPUC}
	We recall that $\varphi(z)=\e^{w(z)}$, $z\in S^1$ with $w(z)$ defined as in~\eqref{eqintro: function w and v} and
 that we defined $D_n \coloneqq \det (T_n(\varphi))$ (by convention $D_{-1}=1$) to be the $n$-Toeplitz determinant associated to the symbol~$\varphi$ (see equations~\eqref{eq:Tn} and~\eqref{def:Dn}). Because $\varphi(z)$ is a real nonnegative function, $D_n \in \mathbb{R}_{>0}$.
	\begin{Proposition}
		If $\varphi(z)$ is a real nonnegative function, we have that
		\begin{equation}
			\label{eq:OPUC in Toeplitz form}
			p_\ell(z)=\frac{1}{\sqrt{D_\ell D_{\ell-1}}}\det\begin{pmatrix}
				\varphi_0 & \varphi_{-1} & \dots & \varphi_{-\ell+1}&\varphi_{-\ell}\\
				\varphi_1& \varphi_0 & \dots & \varphi_{-\ell+2} & \varphi_{-\ell+1}\\
				\vdots& \vdots & \ddots & \vdots &\vdots\\
				\varphi_{\ell-1}& \varphi_{\ell-2}& \dots & \varphi_0 &\varphi_{-1}\\
				1 &z &\dots &z^{\ell-1}& z^\ell
			\end{pmatrix}\!, \qquad
\ell\geq 0.
		\end{equation}
	\end{Proposition}
	\begin{proof}
		The proof is similar to the one for the orthogonal polynomials on the real line, that can be found, e.g., \cite[equation~(3.5)]{DeiftOPRM}, and following discussion.
	\end{proof}

	\begin{Corollary}
		The ratio of two consecutive Toeplitz determinants is expressed as
		\begin{equation}\label{eq: ratio Toeplitz}
			\frac{D_{\ell-1}}{D_{\ell}}=\kappa_\ell^2,\qquad
\ell \geq 0.
		\end{equation}
	\end{Corollary}
	\begin{proof}
		Thanks to formula~\eqref{eq:OPUC in Toeplitz form}, we have that
		\begin{equation*}
			p_\ell(z)= \frac{1}{\sqrt{D_\ell D_{\ell-1}}}\det\begin{pmatrix}
				\varphi_0 & \varphi_{-1} & \dots & \varphi_{-\ell+1}\\
				\varphi_1& \varphi_0 & \dots & \varphi_{-\ell+2} \\
				\vdots& \vdots & \ddots & \vdots \\
				\varphi_{\ell-1}& \varphi_{\ell-2}& \dots & \varphi_0
			\end{pmatrix}z^\ell+\dots =\sqrt{\frac{D_{\ell-1}}{D_\ell}}z^\ell+\cdots,
		\end{equation*}
		and by definition $p_\ell(z)= \kappa_\ell \pi_\ell(z)$ with the latter being the $\ell$-th monic orthogonal polynomial on $S^1$. Thus formula~\eqref{eq: ratio Toeplitz} follows.
	\end{proof}
	
	\subsection{Riemann--Hilbert problem associated to OPUC}
	The family $\{\pi_n\}$ of orthogonal polynomials has a well-known characterization in terms of a $2 \times 2$ dimensional Riemann--Hilbert problem, also depending on $n \geq 0$.
	\begin{problem}\label{RH1}
		The function $Y(z)\coloneqq Y(n,\theta_j;z)\colon\mathbb{C} \rightarrow \mathrm{GL}(2,\mathbb{C})$ has the following properties:
		\begin{enumerate}\itemsep=0pt
			\item[(1)] $Y(z)$ is analytic for every $z\in \mathbb{C}\setminus S^1$;
			\item[(2)] $Y(z)$ has continuous boundary values $Y_{\pm}(z)$ while approaching non-tangentially $S^1$ either from the left or from the right, and they are related for all $z\in S^1$ through
			\begin{equation*}
				Y_+(z)=Y_-(z)J_Y(z),\qquad \text{with}\quad J_Y(z)=\begin{pmatrix}
					1 &z^{-n}\e^{w(z)}\\
					0&1
				\end{pmatrix}\!;
			\end{equation*}
			\item[(3)] $Y(z)$ is normalized at $\infty$ as
			\begin{equation*}
				Y(z) \sim \Bigg(I+\sum_{j=1}^{\infty} \frac{Y_j(n,\theta_j)}{z^j}\Bigg)z^{n\sigma_3},\qquad z \rightarrow \infty,
			\end{equation*}
			where $\sigma_3$ denotes the Pauli's matrix $\sigma_3\coloneqq \left(\begin{smallmatrix}
				1&\hphantom{-} 0\\0&-1
			\end{smallmatrix}\right)$.
		\end{enumerate}	
	\end{problem}
	It is known from~\cite{baik1999distribution} that the above Riemann--Hilbert problem, for each $n\geq 0$, admits a unique solution which is explicitly written in terms of the family $\lbrace\pi_n(z)\rbrace$. Before stating the result, we introduce the following notation. For every polynomial $q(z)$, $z\in \mathbb{C}$, its reverse polynomial $q^*(z)$ is defined as the polynomial of the same degree such that
	\begin{equation*}
		q^*(z)\coloneqq z^n\overline{q\big(\bar{z}^{-1}\big)}.
	\end{equation*}
	For every $\big(L^p\big(S^1\big)\big)$ function $f(y)$, its Cauchy transform $\mathcal{C}f(z)$ is defined for any $z \notin S^1$ as
	\begin{equation*}
		\left(\mathcal{C}f(y)\right)(z) \coloneqq\frac{1}{2\pi \i} \int_{S^1} \frac{f(y)}{y-z}\,\d y.
	\end{equation*}
\begin{Remark}
Notice that the results in~\cite{baik1999distribution} for the Riemann--Hilbert characterization a family of orthogonal polynomials on the unit circle are a sort of extension of the results known from~\cite{FokasItsKitaev2, FokasItsKitaev1} for the case of orthogonal polynomials on the real line.
\end{Remark}
	\begin{Theorem}
		For every $n\geq 0$, the Riemann--Hilbert Problem~$\ref{RH1}$ admits a unique solution~$Y(z)$ that is written as
\begin{equation}\label{eq:sol RH}
Y(z)= \begin{pmatrix}
\pi_n(z)&\mathcal{C}\big(y^{-n}\pi_n(y)\e^{w(y)}\big)(z)
\\[1mm]	
-\kappa^2_{n-1}\pi_{n-1}^*(z)&-\kappa_{n-1}^2\mathcal{C}\big(y^{-n}\pi_{n-1}^*(y)\e^{w(y)}\big)(z)			\end{pmatrix}\!.
		\end{equation}
		Moreover, $\det(Y(z))\equiv 1$.
	\end{Theorem}
	\begin{proof}
		See~\cite[Lemma 4.1]{baik1999distribution}.
	\end{proof}

	The solution $Y(z)$ has a symmetry which will be very useful in the following section.
	\begin{Corollary}
		The unique solution $Y(z)$ of the Riemann--Hilbert Problem~$\ref{RH1}$ is such that
		\begin{align}
			\label{eq:sym Y}
			&Y(z) = \sigma_3Y(0)^{-1}Y\big(z^{-1}\big)z^{n\sigma_3}\sigma_3,\\ \label{eq:sym Yconj}
			& Y(z)=\overline{Y(\bar{z})}.
		\end{align}
	\end{Corollary}
	\begin{proof}
	See~\cite[Proposition 5.12]{BaikDeiftSuidan}.
	\end{proof}

Notice that the factor $Y(0)=Y(n, \theta_j;0)$ appearing in equation~\eqref{eq:sym Y} has a very explicit form by equation~\eqref{eq:sol RH}. This will be useful in the following sections.
	\begin{Lemma}
		For every $n\geq 0$, we have
		\begin{equation}
			\label{eq: Y(0)}
			Y(0)=Y(n,\theta_j;0)=\begin{pmatrix}
				x_n&\kappa_{n}^{-2}\\
				-\kappa^2_{n-1}&x_n
			\end{pmatrix}\!,
		\end{equation}
		where we denoted with $x_n\coloneqq\pi_n(0)$, and $\kappa_n$ is defined as in equation~\eqref{eq: n th opuc}. Moreover, we have
		\begin{equation}
			\label{eq:fundamental rel}
			\frac{\kappa_{n-1}^2}{\kappa_n^2}=1-x_n^2,
		\end{equation}
	and we have $x_n \in \mathbb{R}$.
	\end{Lemma}
	\begin{proof}
		The first column of $Y(n;0)$ directly follows from the evaluation in $z=0$ of $Y(n;z)$ as given in equation~\eqref{eq:sol RH}. Indeed, $Y^{11}(n;0)=\pi_{n}(0)$ and $Y^{21}(n;0)= -\kappa_{n-1}^2\pi^*_{n-1}(0)$ but we observe that
		\begin{equation*}
			\pi^*_{n-1}(0) = z^{n-1}\overline{\pi_{n-1}\big(\bar{z}^{-1}\big)}\big\vert_{z=0} = z^{n-1}\big(z^{-(n-1)}+\dots +\overline{\pi_{n-1}(0)}\big)\big\vert_{z=0}=1.
		\end{equation*}
		Thus we conclude that $Y^{21}(n;0)= -\kappa_{n-1}^2$. For what concerns the second column of $Y(n;0)$, we first find the $(2,2)$-entry. This is indeed easily deduced from the symmetry given in~\eqref{eq:sym Y}. In~the limit for $z\rightarrow \infty$ it gives
		\begin{equation*}
			Y(n;0)=\sigma_3Y^{-1}(n;0)\sigma_3,
		\end{equation*}
		thus $Y^{22}(n;0)=Y^{11}(n;0)=\pi_n(0)$. Finally, for the entry $(1,2)$ of $Y(n;0)$, we compute it explicitly using the orthonormality property of the polynomials $p_m(z)$
		\begin{align*}
Y^{12}(n;0)&= \frac{1}{2\pi {\rm i}}\int_{S^1}\frac{\pi_n(s)s^{-n}w(s)}{s}\,{\rm d}s = \int_{-\pi}^{\pi}\pi_n\big(\e^{{\rm i}\theta}\big)\overline{\e^{{\rm i} n \theta}}w\big(\e^{{\rm i}\theta}\big)\frac{{\rm d}\theta}{2\pi}
\\
&= \frac{1}{\kappa_n^2}\int_{-\pi}^{\pi}p_n\big(\e^{{\rm i}\theta}\big)\overline{p_n\big(\e^{{\rm i}\theta}\big)}w\big(\e^{{\rm i}\theta}\big)\frac{{\rm d}\theta}{2\pi}=\frac{1}{\kappa_n^2}.	
	\end{align*}
		Equation~\eqref{eq:fundamental rel} comes from the fact that $\det(Y(n,\theta_j;z))=1$ identically in $z$ and so in particular for $z=0$ by writing $Y(n,\theta_j;0)$ as in equation~\eqref{eq: Y(0)}, relation~\eqref{eq:fundamental rel} is obtained.
		
		Finally, the fact that $x_n$ is real follows from the entry $(1,1)$ of equation~\eqref{eq:sym Yconj} together with equation~\eqref{eq:sol RH}.
	\end{proof}

At this point, we are already able to express the ratio of Toeplitz determinants in terms of the constant term of the monic orthogonal polynomials, as follows.

	\begin{Corollary}\label{cor:toeplitz recursion relation xn} For every $n\geq 1$, the Toeplitz determinants $D_n$ satisfy the recursion relation
	\begin{equation}\label{eq: discrete log der }
		\frac{D_{n-2}D_{n }}{D_{n-1}^2} = 1-x_n^2.
	\end{equation}
	\end{Corollary}

\begin{proof}
	Putting together equation~\eqref{eq:fundamental rel} with equation~\eqref{eq: ratio Toeplitz} (for two consecutive integers) we obtain the recursion relation~\eqref{eq: discrete log der }.
\end{proof}

We emphasize again that the symbol $\varphi(z)$ actually depends on the natural parameter $N$, so the Toeplitz determinants $D_n$, $n \geq 1$~\eqref{def:Dn} do as well as $x_n=\pi_n(0)$, $n \geq 1$ do (since it is the constant coefficient of the $n$-th monic OPUC w.r.t.\ the $N$-depending measure~\eqref{eq: measure},~\eqref{eqintro: function w and v}). The $N$-dependence of the latter will be emphasized in the following section, where $x_n$ is proved to be a solution of the $N$-th higher order generalization of the discrete Painlev\'e~II equation.
	
	We consider now the following matrix-valued function
	\begin{equation}
		\label{eq:Psi}
		\Psi(n,\theta_j;z)\coloneqq \begin{pmatrix}
			1&0\\
			0& \kappa_{n}^{-2}
		\end{pmatrix}Y(n,\theta_j;z)\begin{pmatrix}
			1&0\\
			0&z^n
		\end{pmatrix}\e^{w(z)\frac{\sigma_3}{2}}.
	\end{equation}
	Thanks to the properties of $Y(z;n,\theta_j)$ from the Riemann--Hilbert Problem~\ref{RH1} one can prove that $\Psi(n,\theta_j;z)$ satisfies the following Riemann--Hilbert problem.

\begin{problem}\label{RH2}
The function $\Psi(z)\coloneqq \Psi(n,\theta_j;z)\colon\mathbb{C} \rightarrow \mathrm{GL}(2,\mathbb{C})$ has the following properties:
\begin{enumerate}\itemsep=0pt
\item[(1)] $\Psi(z)$ is analytic for every $z \in \mathbb{C}\setminus \big\lbrace S^1\cup \lbrace 0 \rbrace \big\rbrace$;

\item[(2)] $\Psi(z)$ has continuous boundary values $ \Psi_{\pm}(z)$ while approaching non-tangentially $S^1$ either from the left or from the right, and they are related for all $z \in S^1$ through
\begin{equation}\label{eq:jump Psi}
\Psi_+ (z) =\Psi_-(z) J_{0}, \qquad
J_{0}=\begin{pmatrix}1 &1 \\0&1\end{pmatrix}\!;
\end{equation}

\item[(3)] $\Psi(z)$ has asymptotic behavior near $0$ given by
\begin{equation}\label{eq:Psi in 0}
\Psi(z) \sim \begin{pmatrix}1&0\\0& \kappa_{n}^{-2}\end{pmatrix}
Y(0)\Bigg(I+\sum_{j= 1}^\infty z^j\widetilde{Y}_j(n)\Bigg)
\begin{pmatrix}1&0\\0&z^n\end{pmatrix}\e^{w(z)\frac{\sigma_3}{2}},\qquad
z\rightarrow 0;
\end{equation}

\item[(4)] $\Psi(z)$ has asymptotic behavior near $\infty$ given by
\begin{equation}\label{eq: Psi in infy}
\Psi(z) \sim \begin{pmatrix}1&0\\0& \kappa_{n}^{-2}\end{pmatrix}
\Bigg(I+\sum_{j= 1}^\infty\frac{Y_j(n)}{z^j}\Bigg)
\begin{pmatrix}z^n&0\\0&1\end{pmatrix}\e^{w(z)\frac{\sigma_3}{2}},\qquad
\vert z \vert \rightarrow \infty.
\end{equation}
\end{enumerate}
\end{problem}

\begin{Proposition}
The function $\Psi(n,\theta_j;z)$ defined in~\eqref{eq:Psi} solves the Riemann--Hilbert Problem~$\ref{RH2}$.
\end{Proposition}

\begin{proof}
		The analyticity condition and the asymptotic expansions at $0$, $\infty$ given in~\eqref{eq:Psi in 0},~\eqref{eq: Psi in infy} follows directly from the definition~\eqref{eq:Psi} and the fact that $Y(z)$ solves the Riemann--Hilbert Problem~\ref{RH1}. Condition~\eqref{eq:jump Psi} follows from direct computation
\begin{align*}
			\Psi(z)_+&=\begin{pmatrix}
				1&0\\
				0& \kappa_{n}^{-2}
			\end{pmatrix}Y_+(z)\begin{pmatrix}
				1&0\\
				0&z^n
			\end{pmatrix}\e^{w(z)\frac{\sigma_3}{2}}=\begin{pmatrix}
				1&0\\
				0& \kappa_{n}^{-2}
			\end{pmatrix}Y_-(z)J_Y(z)\begin{pmatrix}
				1&0\\
				0&z^n
			\end{pmatrix}\e^{w(z)\frac{\sigma_3}{2}}
			\\[2mm]
			&=\Psi_-(z) \begin{pmatrix}
				1&0\\
				0&z^{-n}
			\end{pmatrix}\e^{-w(z)\frac{\sigma_3}{2}}\begin{pmatrix}
				1 &z^{-n}\e^{w(z)}\\
				0&1
			\end{pmatrix}\begin{pmatrix}
				1&0\\
				0&z^n
			\end{pmatrix}\e^{w(z)\frac{\sigma_3}{2}}
=\Psi_-(z)\begin{pmatrix}
				1&1\\
				0&1
			\end{pmatrix}\!.\!\!\!\! \tag*{\qed}
\end{align*}
\renewcommand{\qed}{}
\end{proof}

\subsection[A linear differential system for Psi(z)]{A linear differential system for $\boldsymbol{\Psi(z)}$}
	From the solution of the Riemann--Hilbert Problem~\ref{RH2}, we deduce the following equations (in~the following we omit in $\Psi$ the dependence on $\theta_j$ that should be considered only as parameters and not actual variables like $n$, $z$).
	\begin{Proposition}
	\label{Prop:LaxPair} We have
		\begin{equation}
			\Psi(n+1;z)=U(n;z)\Psi(n;z),\qquad
			\partial_z\Psi(n;z)=T(n;z)\Psi(n;z)
			\label{Lax_pair}
		\end{equation}
		with
		\begin{equation}	\label{U}
			U(n;z):=\begin{pmatrix}
				z+x_nx_{n+1}&-x_{n+1}\\[1mm]
				-\big(1-x_{n+1}^2\big)x_n&1-x_{n+1}^2
			\end{pmatrix}=\sigma_+z+U_0(n),
		\end{equation}
		where $\sigma_+\coloneqq\begin{pmatrix}
				1& 0\\
				0&0
			\end{pmatrix}$ and
		\begin{equation}
			T(n;z):=T_1(n)z^{N-1}+T_2(n)z^{N-2}+\dots+T_{2N+1}(n)z^{-N-1} = \sum_{k=1}^{2N+1}T_kz^{N-k},
			\label{T}
		\end{equation}
		where
		\begin{equation}\label{eq: T1}
			T_1(n)=\frac{\theta_N}{2}\sigma_3.
		\end{equation}
	\end{Proposition}

	\begin{Remark}
	The coefficient $(T_i(n))_{2\leqslant i \leqslant 2N+1}$ defined in equation~\eqref{T} will be computed in Section~\ref{sec3}.
	\end{Remark}

	\begin{proof}
		We first prove the first equation. We start by defining the quantity
\[
U(n;z)\coloneqq \Psi(n+1;z)\Psi^{-1}(n;z).
\]
Since the jump condition for $\Psi(z)$~\eqref{eq:jump Psi} is independent of $n$, $U(n;z)$ is analytic everywhere. Plugging in equation~\eqref{eq: Psi in infy}, we have the expansion at $\infty$
\begin{align*}
U(n;z) = {}&\begin{pmatrix}
			1&0\\
			0& \kappa_{n+1}^{-2}
		\end{pmatrix}\bigg(I+\frac{Y_1(n+1)}{z} + \mathcal{O}\big(z^{-2}\big)\bigg)z^{(n+1)\sigma_3}\begin{pmatrix}1&0\\0&z
		\end{pmatrix}z^{-n\sigma_3}
\\
&\times \bigg(I-\frac{Y_1(n)}{z} + \mathcal{O}\big(z^{-2}\big)\bigg)\begin{pmatrix}
			1&0\\
			0& \kappa_{n}^{2}
		\end{pmatrix}\!,
\end{align*}
from which we deduce that $U(n;z)$ is a polynomial in $z$ of degree $1$, by Liouville theorem. Moreover, its matrix-valued coefficient are written as
		\begin{equation*}
			U(n;z)=z\begin{pmatrix}
				1&0\\
				0&0
			\end{pmatrix}+\underbracket{\begin{pmatrix}
					1&0\\
					0& \kappa_{n+1}^{-2}
				\end{pmatrix}Y(n+1;0)\begin{pmatrix}
					1&0\\0&0
				\end{pmatrix}Y^{-1}(n;0)\begin{pmatrix}
					1&0\\
					0& \kappa_{n}^{2}
			\end{pmatrix}}_{=U_0(n)}.
		\end{equation*}
		Doing the computation and using equation~\eqref{eq: Y(0)}, we obtain
		\begin{align*}
U_0(n) &= \begin{pmatrix}
			Y^{11}(n+1;0)Y^{22}(n;0) &-\kappa_n^2Y^{11}(n+1;0)Y^{12}(n;0)\\
			\kappa_{n+1}^{-2}Y^{21}(n+1;0)Y^{22}(n,0)&-Y^{21}(n+1;0)Y^{12}(n;0)
		\end{pmatrix}
\\
&=\begin{pmatrix}
			x_{n+1}x_n&-x_{n+1}\\
			-\big(1-x_{n+1}^2\big)x_n&1-x_{n+1}^2
		\end{pmatrix}\!.
\end{align*}
		For what concerns the second equation, we define $T(n;z)\coloneqq \partial_z\Psi(n;z)\Psi^{-1}(n;z)$.
		From the asymptotic behavior of $\Psi(n;z)$ at $0$ and $\infty$, we can deduce that $T(n;z)$ is a meromorphic function in $z$ with behavior at $\infty$ described by
		\begin{equation*}
			T(n;z)\sim \begin{pmatrix}
				1&0\\
				0& \kappa_{n}^{-2}
			\end{pmatrix}\bigg(I+\frac{Y_1(n)}{z}+O\big(z^{-2}\big)\bigg)\frac{V'(z)}{2}\sigma_3
\bigg(I-\frac{Y_1(n)}{z}+O\big(z^{-2}\big)\bigg)\begin{pmatrix}
				1&0\\
				0& \kappa_{n}^{2}
			\end{pmatrix}
		\end{equation*}
		(polynomial behavior of degree $N-1$) while at $0$ its behavior is described by
		\begin{gather*}
T(n;z)\sim \begin{pmatrix}
			1&0\\
			0& \kappa_{n}^{-2}
		\end{pmatrix}Y(n,0)\big(I+\tilde{Y}_1(n)z+O\big(z^{2}\big)\big) \\
\hphantom{T(n;z)\sim}{}
\times \frac{-V'(z^{-1})}{2z^2}\sigma_3\big(I-\tilde{Y}_1(n)z+O\big(z^{2}\big)\big)\begin{pmatrix}
			1&0\\
			0& \kappa_{n}^{2}
		\end{pmatrix}\!,
\end{gather*}
		i.e., there is a pole of order $N+1$. In conclusion, we can write
		\begin{equation*}
T(n;z) = \frac{\theta_N}{2}\sigma_3z^{N-1}+T_2(n)z^{N-2}+\dots +T_{2N+1}(n)z^{-N-1}.\tag*{\qed}
\end{equation*}
\renewcommand{\qed}{}
\end{proof}

Moreover, thanks to the symmetry for the solution of the Riemann--Hilbert problem $Y(z)$ stated in~\eqref{eq:sym Y}, we have that the coefficient matrix $T(n;z)$ satisfies a symmetry property.

	\begin{Proposition} $T(n;z)$ has the following symmetry:
		\begin{equation}\label{eq:symmetry T}
			T(n;z^{-1})=-z^2\big(K(n)T(n;z)K(n)^{-1}-nz^{-1}I_2\big)
		\end{equation}
		with $K(n)\coloneqq \begin{pmatrix}
			1&0\\
			0& \kappa_{n}^{-2}
		\end{pmatrix}Y(n;0)\sigma_3\begin{pmatrix}
			1&0\\
			0& \kappa_{n}^{2}
		\end{pmatrix}$.
	\end{Proposition}

	\begin{Remark}
		Notice that for all $n$, the matrix $K(n)$ is s.t.\ $K(n)^{-1}=K(n)$ since we have the identity $x_n^2+\frac{\kappa_{n-1}^2}{\kappa_n^2}=1$.
	\end{Remark}

	\begin{proof}
		On the one hand,
		\begin{equation*}
		\partial_z\big(\Psi\big(n;z^{-1}\big)\big)=-\dfrac{1}{z^2}T\big(n;z^{-1}\big)\Psi\big(n;z^{-1}\big).
		\end{equation*}
		On the other hand, using the symmetry~\eqref{eq:sym Y} for $Y$ we deduce the following symmetry for $\Psi$:
\begin{equation*}
\Psi\big(n;z^{-1}\big)=z^{-n}\begin{pmatrix}
			1&0\\
			0& \kappa_{n}^{-2}
		\end{pmatrix}Y(0)\sigma_3\begin{pmatrix}
			1&0\\
			0& \kappa_{n}^{2}
		\end{pmatrix}\Psi(n;z)\sigma_3.
\end{equation*}
		This previous equation leads to
		\begin{equation*}
		\partial_z\big(\Psi\big(n;z^{-1}\big)\big)=z^{-n}\begin{pmatrix}
			1&0\\
			0& \kappa_{n}^{-2}
		\end{pmatrix}Y(0)\sigma_3\begin{pmatrix}
			1&0\\
			0& \kappa_{n}^{2}
		\end{pmatrix}\partial_z\Psi(n;z)\sigma_3-nz^{-1}\Psi\big(n;z^{-1}\big).
		\end{equation*}
		Then
		\begin{align*}
		T\big(n;z^{-1}\big)={}&-z^2\bigg(\!\begin{pmatrix}
			1&0\\
			0& \kappa_{n}^{-2}
		\end{pmatrix}Y(0)\sigma_3\begin{pmatrix}
			1&0\\
			0& \kappa_{n}^{2}
		\end{pmatrix}T(n;z)\begin{pmatrix}
			1&0\\
			0& \kappa_{n}^{-2}
		\end{pmatrix}\sigma_3Y(0)^{-1}
\\
&\times\begin{pmatrix}
			1&0\\
			0& \kappa_{n}^{2}
		\end{pmatrix}-nz^{-1}I_2\bigg).\tag*{\qed}
\end{align*}
\renewcommand{\qed}{}
\end{proof}

The symmetry~\eqref{eq:symmetry T} reflects on the coefficients $T_k(n)$, $k=1,\dots,2N+1$ as written below.
	\begin{Corollary} The coefficients $T_k(n)$, $k=1,\dots,2N+1$ satisfy
	\begin{align}
		&\label{eq:symmetry Tjs}
		T_j(n)=-K(n)T_{2N+2-j}(n)K(n)^{-1}, \qquad j=1,\dots, N, \\
		&\label{eq:symmetry TN+1}T_{N+1}(n) =-K(n)T_{N+1}(n)K(n)^{-1}+nI_2.
	\end{align}
\end{Corollary}
\begin{proof}
	Indeed, by replacing the exact shape of $T(n;z) $ in equation~\eqref{eq:symmetry T}, we have
	\begin{align*}
\begin{split}
		\sum_{k=1}^{2N+1}T_k(n)z^{-N+k}&=T\big(n;z^{-1}\big)
=-z^2\Bigg(\sum_{k=1}^{2N+1}KT_k(n)K^{-1}z^{N-k}-nz^{-1}I_2\Bigg)\\
		&=-\sum_{k=1}^{2N+1}KT_k(n)K^{-1}z^{N+2-k}+nzI_2\\
		&=-\sum_{j=1}^{2N+1}KT_{2N+2-j}(n)K^{-1}z^{-N+j}+nzI_2,
\end{split}
	\end{align*}
	so looking at the powers $z^{-N+j}$ for $j=1,\dots, N$, we get equation~\eqref{eq:symmetry Tjs} and for $j=N+1$, we get equation~\eqref{eq:symmetry TN+1}.
\end{proof}

Notice first that from equations~\eqref{eq:symmetry Tjs} if the first $N+1$ coefficients of $T(n;z)$ are known, then we can obtain the remaining ones. Second, notice that the coefficient $T_{N+1}(n)$ plays an important role since it solves an equation, the one given in~\eqref{eq:symmetry TN+1}.

\subsection{Relation with the Cresswell--Joshi Lax pair}
To conclude this section, we describe how the Lax pair~\eqref{Lax_pair} is related with the one of the discrete Painlev\'e~II hierarchy~\eqref{eq:lin sys intro} originally introduced by
	Cresswell and Joshi in~\cite{joshi1999discrete} as follows.
	\begin{Definition} 	
		 A Lax pair for the discrete Painlev\'e~II hierarchy is given by a pair of matrices $(L_n(z),M_n(z))$, defining the coefficients of a discrete-differential system for a matrix-valued function $\Phi(n;z)$, such as
		\begin{align}
			\label{eq:JoshiLaxPairL}
			&\Phi(n+1;z)=\begin{pmatrix}
				z & x_n\\
				x_n & 1/z
			\end{pmatrix}\Phi(n;z)=L_n(z)\Phi(n;z),\\
			\label{eq:JoshiLaxPairM}
			&\dfrac{\partial}{\partial z}\Phi(n;z)=M_n(z)\Phi(n;z),
		\end{align}
		with the property that \[ M_n(z)=\begin{pmatrix}
			A_n(z) & B_n(z)\\
			C_n(z) & -A_n(z)
		\end{pmatrix}\]
with $A_n$, $B_n$ and $C_n$ are rational in $z$ (and depending also on $N$).
	\end{Definition}

\begin{Remark}
Specifically, in~\cite[Section 3.1]{joshi1999discrete}, the authors proved that the compatibility condition of the system of equations~\eqref{eq:JoshiLaxPairL} and~\eqref{eq:JoshiLaxPairM} defines the coefficients of the matrix~$M_n(z)$, leaving in turns only one discrete equation of order~$2N$ for~$x_n$. This is defined as the $N$-th member of the discrete Painlev\'e~II hierarchy.
\end{Remark}

	We establish now a link between this Lax Pair and the system~\eqref{Lax_pair} we obtained starting from the OPUC. We define
			\begin{equation*}
		\Phi(n;z):=\sigma_3\begin{pmatrix}
		z^{-n+3/2} & 0\\
		0 & z^{-n+1/2}
	\end{pmatrix}\begin{pmatrix}
		1 & 0\\
		-x_{n-1} & 1
	\end{pmatrix}\Psi\big(n-1;z^2\big).
		\end{equation*}

	\begin{Proposition}
	\label{Prop:connectionJoshi}
		$\Phi(n;z)$ defined as above satisfies the system of equations~\eqref{eq:JoshiLaxPairL} and~\eqref{eq:JoshiLaxPairM}.
	\end{Proposition}

	\begin{proof}
		First we compute the discrete equation for $\Phi(n;z)$. From the definition, we have
		\begin{align*}
			\Phi(n+1;z)=\sigma_3\begin{pmatrix}
				z^{-n+1/2} & 0\\
				0 & z^{-n-1/2}
			\end{pmatrix}\begin{pmatrix}
				1 & 0\\
				-x_{n} & 1
			\end{pmatrix}\Psi\big(n;z^2\big).
		\end{align*}
		According to equation~\eqref{Lax_pair},
		\begin{align*}
			\Phi(n+1;z)={}&\sigma_3\begin{pmatrix}
				z^{-n+1/2} & 0\\
				0 & z^{-n-1/2}
			\end{pmatrix}\begin{pmatrix}
				1 & 0\\
				-x_{n} & 1
			\end{pmatrix}U\big(n-1;z^2\big)\Psi\big(n-1;z^2\big)\\
			={}&\sigma_3\begin{pmatrix}
				z^{-n+1/2} & 0\\
				0 & z^{-n-1/2}
			\end{pmatrix}\begin{pmatrix}
				1 & 0\\
				-x_{n} & 1
			\end{pmatrix}U\big(n-1;z^2\big)\begin{pmatrix}
				1 & 0\\
				x_{n-1} & 1
			\end{pmatrix}
\\
&\times\begin{pmatrix}
				z^{n-3/2} & 0\\
				0 & z^{n-1/2}
			\end{pmatrix}\sigma_3\Phi(n;z)
			=\begin{pmatrix}
				z & x_n\\
				x_n & 1/z
			\end{pmatrix}\Phi(n;z).
		\end{align*}
		Now we compute the derivative with respect to $z$.

		Defining $M_n(z):=\big(\frac{\partial}{\partial z}\Phi(n;z)\big)\Phi(n;z)^{-1}$, similar computations lead to
		\begin{align}
			M_n(z)={}&z^{-1}\sigma_3\begin{pmatrix}
				-n+3/2 &0\\
				0 & -n+1/2
			\end{pmatrix}\sigma_3+2z\sigma_3\begin{pmatrix}
				z &0\\
				0 & 1
			\end{pmatrix}\begin{pmatrix}
				1 & 0\\
				-x_{n-1} & 1
			\end{pmatrix}\nonumber
\\[1mm]
&\times T\big(n-1;z^2\big)\begin{pmatrix}
				1 & 0\\
				x_{n-1} & 1
			\end{pmatrix}\begin{pmatrix}
				z^{-1} &0\\
				0 & 1
			\end{pmatrix}\sigma_3.
			\label{eq:M(n,z)}
		\end{align}
		We need to prove two things: first the trace of $M_n(z)$ is null and then entries of $M_n(z)$ are rational in $z$.

		For the trace of $M_n(z)$ we use the fact that $\Tr(T(n;z))=nz^{-1}$. Then
		\begin{equation*}
		\Tr(M_n(z))=(-2n+2)z^{-1}+2z\Tr\big(T\big(n-1;z^2\big)\big)=0.
		\end{equation*}
		From the expression of $T(n;z)$~\eqref{T} and the equation~\eqref{eq:M(n,z)}, we conclude entries of $M_n(z)$ are rational in $z$.
	\end{proof}
	
	\section{From the Lax Pair to the discrete Painlev\'e~II hierarchy}
	\label{sec3}
	In this section, we study the compatibility condition associated to the linear system~\eqref{Lax_pair}. This first allows us to reconstruct completely the matrix $T(n;z)$ and then to obtain an explicit $2N$ order discrete equation for $x_n$ which corresponds to equation~\eqref{eq:dPIIhierarchy intro}.
	\subsection{The symmetry in the compatibility condition}
	We study the consequences of the symmetry~\eqref{eq:symmetry T} for the matrix $T(n;z)$ on the compatibility condition for the Lax pair introduced in Proposition~\ref{Prop:LaxPair}. More precisely, we show that, thanks to the symmetry~\eqref{eq:symmetry T}, the compatibility condition contains an overdetermined system of equations.

	We recall that the compatibility condition reads as
	\begin{equation}
		\label{ZC}
		\sigma_+-T(n+1;z)U(n;z)+U(n;z)T(n;z)=0,
	\end{equation}
	where we have to replace $U(n;z)$ as in~\eqref{U} and $T(n;z)$ as
	\begin{equation}\label{eq: Tsymmetricform}
		T(n;z) =\sum_{k=1}^{N+1} T_k(n)z^{N-k}+
		\sum_{k=N+2}^{2N+1}-K(n)T_{2N+2-k}(n)K(n)^{-1}z^{N-k},
	\end{equation}
	and with the coefficient $T_{N+1}(n)$ satisfying equation~\eqref{eq:symmetry TN+1}.
	
\begin{Lemma}
The compatibility condition~\eqref{ZC}, for $U(n;z)$, $T(n;z)$ as described above, corresponds to the following system
\begin{align*}
& T_1(n+1)\sigma_+-\sigma_+T_1(n)=0, \\
&T_{j+1}(n+1)\sigma_+-\sigma_+T_{j+1}(n)+T_{j}(n+1)U_0(n)-U_0(n)T_{j}(n)= \sigma_+\delta_{j,N}, \qquad j=1,\dots,N,\\
&T_{N+1}(n) =-K(n)T_{N+1}(n)K(n)^{-1}+nI_2.
\end{align*}
\end{Lemma}

\begin{proof}
	 The compatibility condition~\eqref{ZC}, after replacing $U(n;z)$, $T(n;z)$ of the prescribed form, involves powers of $z$ from $N$ to $-N-1$. Imposing that the coefficients of each of these powers of $z$ is identically zero, we obtain the following equations:
\begin{alignat}{3}
\label{eq: zN}
&z^N\colon&& T_1(n+1)\sigma_+-\sigma_+T_1(n)=0,
\\[1mm]
&z^{N-j},&&\!\!\!\!\! j=1,\dots,N\colon \nonumber
\\
&&&T_{j+1}(n+1)\sigma_+-\sigma_+T_{j+1}(n)+T_{j}(n+1)U_0(n)-U_0(n)T_{j}(n)= \sigma_+\delta_{j,N,}
\label{eq: zj}
\\[1mm]
&z^{-1}\colon&& T_{N+1}(n+1)U_0(n)-U_0(n)T_{N+1}(n)-K(n+1)T_{N}(n+1)K(n+1)^{-1}\sigma_+\nonumber
\\
&&&\qquad{}+\sigma_+K(n)T_{N}(n)K(n)^{-1}=0,
\label{eq: z-1}
\\[1mm]
&z^{N-j},&&\!\!\!\!\! j=N+2,\dots,2N\colon\ \nonumber
\\
&&&-K(n+1)T_{2N+1-j}(n+1)K(n+1)^{-1}\sigma_++\sigma_+K(n)T_{2N+1-j}(n)K(n)^{-1}\nonumber
\\
&&&\qquad{}+U_0(n)K(n)T_{2N+2-j}(n)K(n)^{-1} \nonumber
\\
&&&\qquad{}-K(n+1)T_{2N+2-j}(n+1)K(n+1)^{-1}U_0(n)=0, \label{eq: zN-jbis}
\\[1mm]
&z^{-N-1}\colon\ && -K(n+1)T_{1}(n+1)K(n+1)^{-1}U_0(n)+U_0(n)K(n)T_1(n)K(n)^{-1}=0. \label{eq: z-N-1}
	\end{alignat}
	With the change of indices $2N+1-j=k \iff k=2N+1-j=N-1,\dots,1$, the equation~\eqref{eq: zN-jbis} becomes:
	\begin{align} &-K(n+1)T_k(n+1)K(n+1)^{-1}\sigma_++\sigma_+K(n)T_k(n)K(n)^{-1}\nonumber
\\
&\qquad{}-K(n+1)T_{k+1}(n+1)K(n+1)^{-1}U_0(n)+U_0(n)K(n)T_{k+1}(n)K(n)^{-1}=0. \label{eq: zN-j}
	\end{align}
	 We now show that equations~\eqref{eq: z-1},~\eqref{eq: zN-jbis},~\eqref{eq: z-N-1} are equivalent to the first ones~\eqref{eq: zN},~\eqref{eq: zj} thanks to the symmetry of the coefficients $T_k(n)$ given in~\eqref{eq:symmetry Tjs} together with the equation for $T_{N+1}(n)$, already obtained in~\eqref{eq:symmetry TN+1}.
	
	To start with, we notice the following relations:
	\begin{gather*}
		\widetilde{U}_0(n)\coloneqq K(n+1)^{-1}U_0(n)K(n) = \sigma_+,
	\\
		\widetilde{\sigma}(n) \coloneqq K(n+1)^{-1}\sigma_+K(n)=U_0(n),
	\end{gather*}
	deduced by using multiple times relation~\eqref{eq:fundamental rel}, namely $x_n^2+\frac{\kappa_{n-1}^2}{\kappa_n^2}=1$.
	\begin{enumerate}\itemsep=0pt
		\item[1.] Let us consider first the equation~\eqref{eq: z-N-1} obtained from the coefficient of the term $z^{-N-1}$. Multiplying by $K(n+1)^{-1}$ to the left and by $K(n)$ to the right, we obtain
		\begin{equation*}
			-T_1(n+1)\widetilde{U}_0(n)+\widetilde{U}_0(n)T_1(n)=0,
		\end{equation*}
		that is exactly~\eqref{eq: zN}.
		
		\item[2.] Let us consider now equations~\eqref{eq: zN-j}, obtained from the coefficients of the term $z^{N-j}$, $j =N+2,\dots, 2N$. By multiplying by $K(n+1)^{-1}$ to the left and by $K(n)$ to the right as before, we obtain the equations for $k=N-1,\dots, 1$
		\begin{equation*} -T_k(n+1)\widetilde{\sigma}(n)+\widetilde{\sigma}(n)T_k(n)-T_{k+1}(n+1)\widetilde{U}_0(n) +\widetilde{U}_0(n)T_{k+1}(n)=0,
		\end{equation*}
		which is exactly equation~\eqref{eq: zj} for $j=1,\dots,N-1$.
		\item[3.] The last equation is~\eqref{eq: z-1} obtained from the coefficient of the term $z^{-1}$. We multiply, again, by $K(n+1)^{-1}$ to the left and by $K(n)$ to the right, and we get
		\begin{gather*} K(n+1)^{-1}T_{N+1}(n+1)K(n+1)\widetilde{U}_0(n)-\widetilde{U}_0(n)K(n)^{-1}T_{N+1}(n)K(n)
\\
\qquad{}-T_N(n+1)\widetilde{\sigma}(n)+\widetilde{\sigma}(n)T_N(n)=0,
		\end{gather*}
		and then we replace the symmetry for the term $T_{N+1}(n)$ namely the equation~\eqref{eq:symmetry TN+1} (that indeed it has not be used until now)
		\begin{equation*}
			-T_{N+1}(n+1)\widetilde{U}_0(n)+\widetilde{U}_0(n)T_{N+1}(n) + \widetilde{U}_0(n) -T_N(n+1)\widetilde{\sigma}(n)+\widetilde{\sigma}(n)T_N(n)=0.
		\end{equation*}
		And this is again exactly equation~\eqref{eq: zj}, for $j=N$.
	\end{enumerate}
Thus the compatibility condition~\eqref{ZC} is reduced to the equations in the statement, namely equations~\eqref{eq: zN},~\eqref{eq: zj},~\eqref{eq:symmetry TN+1}.
\end{proof}

Now, we use equations~\eqref{eq: zN},~\eqref{eq: zj} together with the initial condition for $T_1(n)$ given in~\eqref{eq: T1}, to recursively find the coefficients $T_k(n)$, for $k=1,\dots, N+1$, in terms of the $x_{n\pm j}, j=1,\dots,N$. With the coefficients $T_k(n)$ computed in such a way, the symmetry for $T_{N+1}(n)$, i.e., equation~\eqref{eq:symmetry TN+1}, once $T_{N+1}(n)$ is determined, provides an actual discrete equation for~$x_n$ of order~$2N$, that is what we call the higher order analogue of the discrete Painlev\'e~II equation (that coincide for $N=1,2$ to the ones already appeared in~\cite{AV, Borodindiscrete, joshi1999discrete}).

	\subsection{The recursion}
	In this subsection, we explain how equations~\eqref{eq: zN},~\eqref{eq: zj} resulting from the compatibility condition~\eqref{ZC} can be used to find recursively (in $k$) all the coefficients $T_k(n)$, $k=1,\dots, N+1$ of $T(n;z)$.
	\begin{Lemma}
		For every $i=1,\dots, N$, starting from the initial condition~\eqref{eq: T1} $T_1(n)=\frac{\theta_N}{2}\sigma_3$, we have
\begin{gather*}
T_{i+1,12}(n)=x_{n+1}\big(2\Delta^{-1}+I\big)\bigg(\dfrac{x_{n+1}}{v_{n+1}} T_{i,21}(n+1)-x_nT_{i,12}(n)\bigg)+v_{n+1}T_{i,12}(n+1)
\\
\hphantom{T_{i+1,12}(n)=}-x_nx_{n+1}T_{i,12}(n),
\\
T_{i+1,21}(n+1)=x_nv_{n+1}\big(2\Delta^{-1}+I\big)\bigg(\dfrac{x_{n+1}}{v_{n+1}} T_{i,21}(n+1)-x_nT_{i,12}(n)\bigg)+v_{n+1}T_{i,21}(n)
\\
\hphantom{T_{i+1,21}(n+1)=}-x_nx_{n+1}T_{i,21}(n+1),
\\
T_{i+1,11}(n)=-T_{i+1,22}(n)+n\delta_{i,N}=\Delta^{-1}\bigg(\dfrac{-x_{n+1}}{v_{n+1}} T_{i+1,21}(n+1)+x_nT_{i+1,12}(n)\bigg)+n\delta_{i,N},
\end{gather*}
	where 	
	\begin{gather}
	\label{def delta} \Delta\colon\ T_i(n)\to T_i(n+1)-T_i(n),
\\
	\label{def vn} v_n\coloneqq 1-x_{n}^2,
	\end{gather}
	\end{Lemma}

	\begin{proof}
	We rewrite equations~\eqref{eq: zN},~\eqref{eq: zj} for $ i=1,\dots,N$, entry by entry.
	For the first one, we~have
			\begin{equation*}
	\begin{cases}
		T_{1,11}(n+1)-T_{1,11}(n)=0,\\
		T_{1,12}(n)=T_{1,21}(n+1)=0.
	\end{cases}
		\end{equation*}
	This is satisfied by $T_1(n) $ given in~\eqref{eq: T1}. For the second one, for any $1\leqslant i\leqslant N$ we have the four equations:
\begin{gather*}
T_{i+1,11}(n+1)-T_{i+1,11}(n)=-T_{i,11}(n+1)x_n x_{n+1} +T_{i,12}(n+1)\big(1-x_{n+1}^2\big)x_n
\\ \hphantom{T_{i+1,11}(n+1)-T_{i+1,11}(n)=}
{}+x_nx_{n+1}T_{i,11}(n)-x_{n+1}T_{i,21}(n)+\delta_{i,N},
\\[1mm]		T_{i+1,12}(n)=-x_{n+1}T_{i,11}(n+1)+T_{i,12}(n+1)\big(1-x_{n+1}^2\big)-x_nx_{n+1}T_{i,12}(n)+x_{n+1}T_{i,22}(n),
\\[1mm]	
T_{i+1,21}(n+1)=-T_{i,21}(n+1)x_nx_{n+1}+T_{i,22}(n+1)x_n\big(1-x_{n+1}^2\big)
\\\hphantom{T_{i+1,21}(n+1)=}
{}-T_{i,11}(n)x_n\big(1-x_{n+1}^2\big)+\big(1-x_{n+1}^2\big)T_{i,21}(n),
\\[1mm]	
0=T_{i,21}(n+1)x_{n+1}\!-T_{i,22}(n+1)\big(1\!-x_{n+1}^2\big)\!-x_n\big(1\!-x_{n+1}^2\big)T_{i,12}(n)+T_{i,22}(n)\big(1\!-x_{n+1}^2\big).		\label{EqRecTi22}
\end{gather*}

	Using the notations introduced in~\eqref{def delta},~\eqref{def vn}, the previous equations
	with $1\leqslant i\leqslant N$ become
	\begin{gather}
		\Delta T_{i+1,11}(n)=-x_nx_{n+1}\Delta T_{i,11}(n) +x_nv_{n+1}T_{i,12}(n+1)-x_{n+1}T_{i,21}(n)+\delta_{i,N},
		\label{EqRecTi11mod}
	\\[1mm] T_{i+1,12}(n)=-x_{n+1}T_{i,11}(n+1)\!+v_{n+1}T_{i,12}(n+1)\!-x_nx_{n+1}T_{i,12}(n)\!+x_{n+1}T_{i,22}(n),
		\label{EqRecTi12mod}
	\\[1mm] T_{i+1,21}(n+1)=-x_nx_{n+1}T_{i,21}(n+1)+x_nv_{n+1}T_{i,22}(n+1)-x_nv_{n+1}T_{i,11}(n)\nonumber
\\ \hphantom{T_{i+1,21}(n+1)=}
{}+v_{n+1}T_{i,21}(n),
		\label{EqRecTi21mod}
	\\[1mm]
		v_{n+1}\Delta T_{i,22}(n)=x_{n+1}T_{i,21}(n+1)-x_nv_{n+1}T_{i,12}(n).
		\label{EqRecTi22mod}
	\end{gather}
	From these equations, we see that in order to obtain the diagonal terms, there is a ``discrete integration'' to perform, while the off-diagonal terms are directly determined from the previous ones.
	Moreover, we can rewrite the four equation as only two equations involving only the off-diagonal terms. Indeed,
	because of $\Tr(T(n;z))=nz^{-1}$, $T_{i,11}(n,z)=-T_{i,22}(n,z)$ for $1\leqslant i \leqslant N$. Thus~\eqref{EqRecTi22mod} can be written as
\begin{equation*}
v_{n+1}\Delta T_{i,11}(n)=-x_{n+1}T_{i,21}(n+1)+x_nv_{n+1}T_{i,12}(n).
\end{equation*}
	Formally, $1\leqslant i\leqslant N$,
\begin{equation}
	\label{expr:Ti11expl}
	T_{i,11}(n)=-T_{i,22}(n)=\Delta^{-1}\bigg(\dfrac{-x_{n+1}}{v_{n+1}} T_{i,21}(n+1)+x_nT_{i,12}(n)\bigg),
\end{equation}
which still holds for $i=N+1$ up to adding the ``constant'' $n$ on the right hand side.
Using this in~\eqref{EqRecTi12mod} and~\eqref{EqRecTi21mod}, we obtain:
\begin{gather*}
T_{i+1,12}(n)=x_{n+1}\big(2\Delta^{-1}+I\big)\bigg(\dfrac{x_{n+1}}{v_{n+1}} T_{i,21}(n+1)-x_nT_{i,12}(n)\bigg)+v_{n+1}T_{i,12}(n+1)
\\ \hphantom{T_{i+1,12}(n)=}
{}-x_nx_{n+1}T_{i,12}(n),
\\
T_{i+1,21}(n+1)=x_nv_{n+1}\big(2\Delta^{-1}+I\big)\bigg(\dfrac{x_{n+1}}{v_{n+1}} T_{i,21}(n+1)-x_nT_{i,12}(n)\bigg)+v_{n+1}T_{i,21}(n)
\\ \hphantom{T_{i+1,21}(n+1)=}
-x_nx_{n+1}T_{i,21}(n+1).\tag*{\qed}
\end{gather*}
\renewcommand{\qed}{}
\end{proof}

We notice that, defining the discrete recursion operator
	\begin{gather}
		\label{L}
		\mathcal{L}\begin{pmatrix}u_n\\y_n\end{pmatrix}
=\begin{pmatrix}
			x_{n+1}\big(2\Delta^{-1}+I\big)\bigg(\dfrac{x_{n+1}}{v_{n+1}} y_n-x_nu_n\bigg)+(v_{n+1}(\Delta+I)-x_nx_{n+1})u_n
\\[4mm]
			x_nv_{n+1}\big(2\Delta^{-1}\!+\!I\big)\bigg(\dfrac{x_{n+1}}{v_{n+1}} y_n\!-\!x_nu_n\bigg)\!+\!\big(v_{n+1}(\Delta+I)^{-1}\!-\!x_nx_{n+1}\big)y_n
		\end{pmatrix}\!,
	\end{gather}
	we can rewrite the two equations for the off-diagonal entries of $T_i(n)$ obtained above as
	\begin{equation}
	\label{LTiTi+1}
	\begin{pmatrix}
		T_{i+1,12}(n)\\
		T_{i+1,21}(n+1)
	\end{pmatrix}=\mathcal{L}\begin{pmatrix}
		T_{i,12}(n)\\
		T_{i,21}(n+1)
	\end{pmatrix}\!,\qquad 1\leqslant i\leqslant N.
	\end{equation}
	And, recursively we obtain
	\begin{equation}
		\label{TN+11221}
		\begin{pmatrix}
			T_{N+1,12}(n)\\
			T_{N+1,21}(n+1)
		\end{pmatrix}={\mathcal{L}}^N\begin{pmatrix}
			0\\
			0
		\end{pmatrix}\!.
	\end{equation}
	This procedure allows to construct the whole matrix $T(n;z)$, starting from the initial condition $T_1(n)=\frac{\theta_N}{2}\sigma_3$ and iterating the operator $\mathcal{L}$ we obtain off diagonal terms of $T(n;z)$ and compute diagonal one with equation~\eqref{expr:Ti11expl}.
	Below we implemented this method to find the matrix $T(n;z)$ in the first few cases $N=1,2$.
	\begin{Example} In the case $N=1$, the matrix $T(n;z)=T_1(n)+T_2(n)z^{-1}+T_3(n)z^{-2}$. Knowing $T_1(n)$, we only have to find $T_2(n)$ using the recurrence relation given from the compatibility, i.e., equations~\eqref{EqRecTi11mod},~\eqref{EqRecTi12mod},~\eqref{EqRecTi21mod} for $i=1$.
		Since: $T_{1,12}(n)=T_{1,21}(n)=0$, and $T_{1,11}(n)=\theta_N/2= -T_{1,22}(n)$, we have
\begin{gather*}
T_{2,11}(n)=n,
\\
T_{2,12}(n)=-x_{n+1}(T_{1,11}(n+1)+T_{1,11}(n))=-\theta_1x_{n+1},
\\
T_{2,21}(n+1)=x_nv_{n+1}(T_{1,22}(n+1)+T_{1,22}(n))=-\theta_1 x_nv_{n+1},
\end{gather*}
and $T_{2,22}(n)=n-T_{2,11}(n)=0$.
		Moreover, the symmetry which reflects terms of $T(n;z)$ two by two gives $T_3(n)=-K(n)T_1(n)K(n)$. Thus the Lax matrix for $N=1$ is
		\begin{equation*}
			T(n;z) = \dfrac{\theta_1}{2}\begin{pmatrix}
				1&\hphantom{-}0\\0&-1
			\end{pmatrix} + \frac{1}{z}\begin{pmatrix}
				n&-\theta_1x_{n+1}\\
				-\theta_1v_{n}x_{n-1}&0
			\end{pmatrix}+\frac{\theta_1}{z^2}\begin{pmatrix}
				\frac{1}{2}-x_n^2&x_n\\
				v_{n}x_n&x_n^2-\frac{1}{2}
			\end{pmatrix}\!.
		\end{equation*}
	\end{Example}

	\begin{Example}
In the case $N=2$, the matrix $T(n;z)=T_1(n)z+T_2(n)+T_3(n)z^{-1}+T_4(n)z^{-2}+T_5(n)z^{-3}$. This time we have to find $T_2(n)$ (that will be almost the same as before) and also~$T_3(n)$ using the recurrence relation given from the compatibility, i.e., equations~\eqref{EqRecTi11mod},~\eqref{EqRecTi12mod},~\eqref{EqRecTi21mod} for $i=1$ and $2$. First we find $T_2(n)$ ($i=1$ above), we have
\begin{gather*}
T_{2,11}(n)=\frac{\theta_1}{2},
\\
T_{2,12}(n)=-x_{n+1}(T_{1,11}(n+1)+T_{1,11}(n))=-\theta_2x_{n+1},
\\
T_{2,21}(n+1)=x_nv_{n+1}(T_{1,22}(n+1)+T_{1,22}(n))=-\theta_2 x_nv_{n+1},
\end{gather*}
and $T_{2,22}(n)=-T_{2,11}=-\frac{\theta_1}{2}$.

		Then we consider the equation for $i=2$ and find $T_3(n)$. We have
\begin{gather*}
\Delta T_{3,11}(n)=x_nv_{n+1}(-\theta_2x_{n+2}) -x_{n+1}(-\theta_2 x_{n-1}v_{n})+1\! \implies \! T_{3,11}(n)=n-\theta_2 x_{n-1}x_{n+1}v_n,
\\
T_{3,12}(n)=-\theta_1x_{n+1}-\theta_2\big(v_{n+1}x_{n+2}-x_nx_{n+1}^2\big),
\\
T_{3,21}(n+1)=\bigl(-\theta_1x_{n}-\theta_2\big(v_nx_{n-1}-x_n^2x_{n+1}\big)\bigr)v_{n+1},
\\
T_{3,22}(n)= n-T_{3,11}(n)=\theta_2 x_{n-1}x_{n+1}v_n.
\end{gather*}
Finally, we take $T_4(n)=-K(n)T_2(n)K(n)$ and $T_5(n)=-K(n)T_1(n)K(n)$. Thus the Lax matrix for $N=2$ is
\begin{align*}
&T(n;z)=z\frac{\theta_2}{2}\begin{pmatrix}
				1&0\\0&-1
			\end{pmatrix}+\begin{pmatrix}
				\frac{\theta_1}{2}&-\theta_2x_{n+1}\\
				-\theta_2 x_{n-1}v_{n}&-\frac{\theta_1}{2}
			\end{pmatrix}
\\
&+\frac{1}{z}\begin{pmatrix}
				n-\theta_2 x_{n-1}x_{n+1}v_n & -\theta_1x_{n+1}-\theta_2\big(v_{n+1}x_{n+2}-x_nx_{n+1}^2\big)\\
				\bigl(-\theta_1x_{n-1}-\theta_2\big(v_{n-1}x_{n-2}-x_nx_{n-1}^2\big)\bigr)v_{n} & \theta_2 x_{n-1}x_{n+1}v_n
			\end{pmatrix}
\\
			&+\frac{1}{z^2}\begin{pmatrix}
				-\theta_2v_n(x_nx_{n-1}+x_nx_{n+1})+\frac{\theta_1}{2}\big(v_n-x_n^2\big)	& -\theta_2\big(v_nx_{n-1}+x_n^2x_{n+1}\big)\\
				-\theta_2\big(v_nx_{n+1}+x_n^2x_{n-1}\big)v_n &\theta_2v_n(x_nx_{n-1}+x_nx_{n+1})-\frac{\theta_1}{2}\big(v_n-x_n^2\big)
			\end{pmatrix}\\&+
			\frac{\theta_2}{z^3}\begin{pmatrix}
				\frac{1}{2}-x_n^2&x_n\\
				v_{n}x_n&x_n^2-\frac{1}{2}
\end{pmatrix}\!.
\end{align*}
\end{Example}

	 Now that we have reconstructed the whole matrix $T(n;z)$ in terms of $x_{n\pm j}$, $j=-N,\dots,N$ we are left with the equation that $T_{N+1}(n)$ has to satisfy, namely~\eqref{eq:symmetry TN+1}. We now show that actually this coincide with only one scalar equation in $T_{N+1, 12}$ and $T_{N+1, 21}$.
	 Indeed, entry by entry it reads as the following system of four equations. From the off-diagonal entries
\begin{align}
\label{eq: hiddendPII 1}
\begin{split}
&v_{n}T_{N+1,12}(n)=x_n(T_{N+1,11}(n)-T_{N+1,22}(n))-T_{N+1,21}(n),
\\	
&v_{n}T_{N+1,21}(n)=x_nv_{n}(T_{N+1,11}(n)-T_{N+1,22}(n))-v_{n}^2T_{N+1,12}(n)
\end{split}
\end{align}
and from the diagonal entries
\begin{align} &n-\big(1+x_n^2\big)T_{N+1,11}(n)-v_{n}T_{N+1,22}(n)+x_nT_{N+1,21}(n)+x_nv_{n}T_{N+1,12}(n)=0,\nonumber\\
		&n-\big(1+x_n^2\big)T_{N+1,22}(n)-v_{n}T_{N+1,11}(n)-x_nT_{N+1,21}(n)-x_nv_{n}T_{N+1,12}(n)=0.\nonumber
	\end{align}
	We notice first that the four above equations are all the same. The first and the second equations are the same up to a multiplication by $v_n$. Using the relation $T_{N+1,11}(n)+T_{N+1,22}(n)=n$, we can rewrite the third and the forth equations and obtain the same equation up to a sign. Finally, multiplying by $x_n$ the first equation and using the relation $T_{N+1,11}(n)+T_{N+1,22}(n)=n$ we obtain the third one. Thus from now on we will refer only to~\eqref{eq: hiddendPII 1}, as for the remaining equation.
	
	Using equation~\eqref{EqRecTi22mod} and $\Tr(T(n;z))=nz^{-1}$, we express equation~\eqref{eq: hiddendPII 1} in function of $T_{N+1,12}(n)$ and $T_{N+1,21}(n)$.
	Consider equation~\eqref{eq: hiddendPII 1}, with the identity $\Tr(T_{N+1}(n))=n$, it is rewritten as
\begin{equation*}
v_nT_{N+1,12}(n)=x_n(n-2T_{N+1,22}(n))-T_{N+1,21}(n).
\end{equation*}
Equation~\eqref{EqRecTi22mod} holds also for $i=N+1$. It means it is possible to replace $T_{N+1,22}(n)$ in the previous equation and obtain
\begin{gather}
\label{eq:DPIINTN+112et21} nx_n-v_nT_{N+1,12}(n)-T_{N+1,21}(n)\nonumber
\\ \qquad
{}-2x_n\Delta^{-1}\biggl(-x_nT_{N+1,12}(n)+\dfrac{x_{n+1}}{v_{n+1}}(\Delta+I)T_{N+1,21}(n)\biggr)=0.
\end{gather}
		
\subsection[The relation between T\_\{i,12\}(n) and T\_\{i,21\}(n)]{The relation between $\boldsymbol{T_{i,12}(n)}$ and $\boldsymbol{T_{i,21}(n)}$}

	The previous equation~\eqref{eq:DPIINTN+112et21} depends on $T_{N+1,12}(n)$ and $T_{N+1,21}(n)$. The aim of this part is to establish a connection between $T_{i,12}(n)$ and $T_{i,21}(n)$ to rewrite equation~\eqref{eq:DPIINTN+112et21} just in function of~$T_{N+1,12}(n)$.

To accomplish this, we study the compatibility condition of $C(n;z)\coloneqq T(n;z)^2$ and $U(n;z)$.
$C(n;z)$ is rational in $z$ with a pole of order $-2N-2$ at $0$. We write $C(n;z)$ as
	\begin{equation}
		C(n;z)=\sum_{i=1}^{4N+1}C_i(n)z^{2N-1-i}
		\label{exprC}
	\end{equation}
	with{\samepage
	\begin{equation}
		C_i(n)\coloneqq\sum_{j=1}^iT_j(n)T_{i+1-j}(n)
		\label{DefCi}
	\end{equation}
	where $C_1(n)=\frac{\theta_N^2}{4}I_2$.}
	
	In what follows we will need the following lemma:
	\begin{Lemma}
		\label{lemma:CiCN+1}
		Diagonal coefficients of $C_i(n)$ defined as in~\eqref{DefCi} satisfy the following equation:
\begin{gather*}
\forall 1\leqslant i\leqslant N,\qquad C_{i,11}(n)=C_{i,22}(n),
\\
C_{N+1,11}(n)=n\theta_N+C_{N+1,22}(n).
\end{gather*}
	\end{Lemma}

	\begin{proof}
		We express $C_{i,11}(n)$ in function of $T_{i,kj}(n)$. With the equation~\eqref{DefCi}
\begin{equation*}
C_{i,11}(n)=\sum_{j=1}^iT_{j,11}(n)T_{i+1-j,11}(n)+T_{j,12}(n)T_{i+1-j,21}(n).
\end{equation*}
Then, the sum index change $j=i-k+1$ leads to
\begin{equation*}
C_{i,11}(n)=\sum_{k=1}^iT_{i-k+1,11}(n)T_{k,11}(n)+T_{i-k+1,12}(n)T_{k,21}(n).
\end{equation*}
		Finally, with the relation $\Tr(T(n;z))=nz^{-1}$,
\begin{itemize}\itemsep=0pt
		\item if $1\leqslant i\leqslant N$,
\begin{equation*}
C_{i,11}(n)=\sum_{k=1}^iT_{i-k+1,22}(n)T_{k,22}(n)+T_{k,21}(n)T_{i-k+1,12}(n)=C_{i,22}(n).
\end{equation*}
		\item if $i=N+1$,
\begin{align*}
C_{N+1,11}(n)&=-2nT_{1,22}(n)+\sum_{k=1}^{N+1}T_{N-k+2,22}(n)T_{k,22}(n)+T_{k,21}(n)T_{N-k+2,12}(n)
\\
&=n\theta_N+C_{N+1,22}(n).
 \tag*{\qed}
\end{align*}
\end{itemize}\renewcommand{\qed}{}
\end{proof}
	
	We deduce the compatibility condition for $C$ and $U$ from the one for $T$ and $U$.
	\begin{Lemma}
		$C(n;z)$~\eqref{exprC} and $U(n;z)$~\eqref{U} satisfy the following compatibility condition:
		\begin{equation}
			C(n+1;z)U(n;z)-U(n;z)C(n;z)=T(n+1;z)\sigma_++\sigma_+T(n;z).
			\label{ZCUC}
		\end{equation}
	\end{Lemma}
	\begin{proof}
		Multiplying on the left (resp.\ on the right) equation~\eqref{ZC} by $T(n+1;z)$ (resp.\ $T(n;z)$) and summing these two equations leads to the result.
	\end{proof}

The left (resp.\ right) hand side of the equation in the previous lemma is an expression in powers of $z$ from $z^{2N-1}$ to $z^{-2N-2}$ (resp.\ from $z^{N-1}$ to $z^{-N-1}$). This equation leads to recursive equation for $C_i(n)$. We consider only expression in powers of $z$ from $z^{2N-1}$ to $z^{N-1}$.

	According to~\eqref{ZC} and~\eqref{ZCUC}, $\forall 1\leqslant i\leqslant N$, $C_i(n)$ and $T_i(n)$ satisfy the same recursive equation (see equations~\eqref{EqRecTi11mod}--\eqref{EqRecTi22mod}). For $i=N+1$, the equation is a bit different. The term with $\delta_{i,N}$ is now multiplied by $\theta_N$.

	From these equations we deduce the following result.
	\begin{Proposition}
		\label{propCiHom}
		Let $C_i(n)$ be as in~\eqref{DefCi}. Then $\forall 1\leqslant i\leqslant N$,
\begin{equation*}
C_i(n)=\alpha_iI_2\qquad \text{and}\qquad
C_{N+1}(n)=\theta_Nn\sigma_++\alpha_{N+1}I_2.
\end{equation*}
\end{Proposition}

	\begin{proof}
		We prove Proposition~\ref{propCiHom} by induction.
		For $i=1$, we already know $C_1(n)=\frac{\theta_N^2}{4}$.
		Suppose $C_i(n)=\alpha_iI_2$ for $i\leqslant N-1$.
		$C_{i+1}(n)$ satisfies the following equations:
		\begin{gather*}
			\Delta C_{i+1,11}(n)=-x_nx_{n+1}\Delta C_{i,11}(n) +x_nv_{n+1}C_{i,12}(n+1)-x_{n+1}C_{i,21}(n)+\theta_N\delta_{i,N},
\\
C_{i+1,12}(n)=-x_{n+1}C_{i,11}(n+1)+v_{n+1}C_{i,12}(n+1)-x_nx_{n+1}C_{i,12}(n)+x_{n+1}C_{i,22}(n),
\\		
C_{i+1,21}(n+1)=-x_nx_{n+1}C_{i,21}(n+1)+x_nv_{n+1}C_{i,22}(n+1)-x_nv_{n+1}C_{i,11}(n)
\\ \hphantom{C_{i+1,21}(n+1)=}
{}+v_{n+1}C_{i,21}(n).
		\end{gather*}
		Using induction hypothesis,
		\begin{gather*}
			\Delta C_{i+1,11}(n)=-0 \cdot x_nx_{n+1} +0\cdot x_nv_{n+1}-0\cdot x_{n+1}+\theta_N\delta_{i,N}=\theta_N\delta_{i,N}, \\
			C_{i+1,12}(n)=-x_{n+1}\alpha_i+0\cdot v_{n+1}-0\cdot x_nx_{n+1}+x_{n+1}\alpha_i=0, \\
			C_{i+1,21}(n+1)=-0\cdot x_nx_{n+1}+x_nv_{n+1}\alpha_i-x_nv_{n+1}\alpha_i+0\cdot v_{n+1}=0.
		\end{gather*}
		From the first equation, we conclude $C_{i+1,11}(n)=\alpha_{i+1}$ if $i\leqslant N-1$ (resp.\ $C_{N+1,11}(n)=\theta_Nn+\alpha_{N+1}$ if $i=N$) and according to Lemma~\ref{lemma:CiCN+1} $C_{i+1,22}(n)=\alpha_{i+1}$ (resp.\ $C_{N+1,22}(n)=\alpha_{N+1}$) which concludes the proof.
\end{proof}

	From equation~\eqref{DefCi} and Proposition~\ref{propCiHom}, we obtain
\begin{gather}
\theta_NT_{i,11}(n)=\alpha_i-\sum_{j=2}^{i-1}T_{j,11}(n)T_{i-j+1,11}(n)+T_{j,12}(n)T_{i-j+1,21}(n),
\label{Ti11}
\\
\theta_NT_{N+1,11}(n)=n\theta_N+\alpha_{N+1}-\sum_{j=2}^{N}T_{j,11}(n)T_{N-j+2,11}(n)
+T_{j,12}(n)T_{N-j+2,21}(n).
\label{TN+111}
\end{gather}
With all this discussion on $C(n;z)$ it is now possible to prove the following proposition.

\begin{Proposition}
\label{prop:symTi1221}
The following holds:
$\forall 1\leqslant i \leqslant N+1$, $T_{i,11}(n)$, $T_{i,12}(n)$ and $T_{i,21}(n)$ are polynomials in $x_{n+j}$'s. Moreover, the following symmetries hold:
\begin{equation*}
\exists (Q_{i,n}((u_{n+j})_{1-i\leqslant j\leqslant i-1}),P_{i,n}((u_{n+j})_{1-i\leqslant j\leqslant i-1}))
\end{equation*}
polynomials in $u_{n+j}$'s such that,
\begin{gather*}
T_{i,11}(n)=Q_{i,n}((x_{n+j})_{1-i\leqslant j\leqslant i-1})
=Q_{i,n}((x_{n-j})_{1-i\leqslant j\leqslant i-1}),
\\
T_{i,12}(n)=P_{i,n}((x_{n+j})_{1-i\leqslant j\leqslant i-1}),
\\
T_{i,21}(n)=v_nP_{i,n}((x_{n-j})_{1-i\leqslant j\leqslant i-1}).
\end{gather*}
\end{Proposition}

\begin{proof}
		We prove this proposition by strong induction.
For $i=1$, $T_1(n)=\frac{\theta_N}{2}\sigma_3$, then defining $Q_{1,n}(u_n):=\frac{\theta_N}{2}$, $P_{1,n}(u_n):=0$; $T_{1,11}(n)=Q_{1,n}(x_n)$, $T_{1,12}(n)=P_{1,n}(x_n)$ and $T_{1,21}(n)=v_nP_{1,n}(x_n)$.

		Now suppose the property true for all $j\in[[1,i]]$ with $i\leqslant N$ and let $(Q_{j,n},P_{j,n})_{j\leqslant i}$ be polynomials in $x_{n+j}$'s satisfying the property.
		According to~\eqref{Ti11} (and~\eqref{TN+111} for $i=N$) and strong induction hypothesis, $T_{i+1}(n)$ is a polynomial in $x_{n+j}$'s and the invariance when you exchange $x_{n+j}$ by $x_{n-j}$ holds.

		Because of equation~\eqref{EqRecTi12mod} (resp.\ equation~\eqref{EqRecTi21mod}) and of induction hypothesis, there exists $P_{i+1,n}((u_{n+j})_{-i\leqslant j\leqslant i})$ (resp.\ $\tilde{P}_{i+1,n}((u_{n+j})_{-i\leqslant j\leqslant i})$) a polynomial such that
\begin{equation*}
T_{i+1,12}(n)=P_{i+1,n}((x_{n+j})_{-i\leqslant j\leqslant i}),
\end{equation*}
respectively,
\begin{equation*}
T_{i+1,21}(n)=\tilde{P}_{i+1,n}((x_{n+j})_{-i\leqslant j\leqslant i}).
\end{equation*}
		Now we establish the link between $P_{i+1,n}$ and $\tilde{P}_{i+1,n}$. According to equation~\eqref{EqRecTi12mod} and the relation $\Tr(T(n;z))=nz^{-1}$,
\begin{align*}
P_{i+1,n}\big((x_{n+j})_{j=-i}^{i}\big)={}&-x_{n+1}Q_{i,n+1}\big((x_{n+j})_{j=-i}^{i-2}\big)
+v_{n+1}P_{i,n+1}\big((x_{n+j})_{j=-i}^{i-2}\big)
\\
&-x_nx_{n+1}P_{i,n}\big((x_{n+j})_{j=1-i}^{i-1}\big)
-x_{n+1}Q_{i,n}\big((x_{n+j})_{j=1-i}^{i-1}\big).
\end{align*}
		Then
\begin{align*} v_nP_{i+1,n}\big((x_{n-j})_{j=-i}^{i}\big)={}&v_n\big(-x_{n-1}Q_{i,n-1}\big((x_{n-j})_{j=-i}^{i-2}\big)
+v_{n-1}P_{i,n-1}\big((x_{n-j})_{j=-i}^{i-2}\big)
\\
&-x_nx_{n-1}P_{i,n}\big((x_{n-j})_{j=1-i}^{i-1}\big)-x_{n-1}Q_{i,n}\big((x_{n-j})_{j=1-i}^{i-1}\big)\big).
		\end{align*}
		From induction hypothesis and $\Tr(T(n;z))=nz^{-1}$
\begin{align*} v_nP_{i+1,n}\big((x_{n-j})_{j=-i}^{i}\big)={}&-x_{n-1}v_nT_{i,11}(n-1)+v_nT_{i,21}(n-1)+x_{n-1}x_nT_{i,21}(n)
\\
&+x_{n-1}v_nT_{i,22}(n).
\end{align*}
		According to equation~\eqref{EqRecTi21mod},
\begin{equation*}
v_nP_{i+1,n}\big((x_{n-j})_{j=-i}^{i}\big)=T_{i+1,21}(n+1).
\end{equation*}
		Then
\begin{equation*}
v_nP_{i+1,n}\big((x_{n-j})_{j=-i}^{i}\big)
=\tilde{P}_{i+1,n}\big((x_{n+j})_{-i\leqslant j\leqslant i}\big)
\end{equation*}		
and this concludes the proof.
\end{proof}

Define $\C[(x_j)_{j\in[[0,2n]]}]$ and the transformation
\begin{align*}
{\rm Perm}_n\colon\quad \C[(x_j)_{j\in[[0,2n]]}]&\longrightarrow\C[(x_j)_{j\in[[0,2n]]}],
\\
P((x_{n+j})_{-n\leqslant j\leqslant n})&\longmapsto P((x_{n-j})_{-n\leqslant j\leqslant n}).
\end{align*}
	From the previous proposition,
\begin{equation}
		 T_{i,21}(n)=v_n{\rm Perm}_n(T_{i,12}(n)).
		 \label{symTi1221}
\end{equation}

	\begin{Remark}
		As a consequence of the Proposition~\ref{prop:symTi1221}, the equation~\eqref{eq: hiddendPII 1} is a polynomial in $x_{n+j}$'s and is invariant when you apply ${\rm Perm}_n$ to this equation because ${\rm Perm}_n^2={\rm Id}$ and ${\rm Perm}_nv_n=v_n{\rm Perm}_n$.
	\end{Remark}
	We use the link we established in Proposition~\ref{prop:symTi1221} between $T_{i,12}(n)$ and $T_{i,21}(n)$ to rewrite the operator $\mathcal{L}$~\eqref{L} as a scalar operator:
	\begin{equation}
		L(u_n)\coloneqq \big(x_{n+1}\big(2\Delta^{-1}+I\big)((\Delta+I)x_{n}{\rm Perm}_n-x_n)+v_{n+1}(\Delta+I)-x_nx_{n+1}\big)u_n.
		\label{Ln}
	\end{equation}
Finally, collecting all the results from the previous sections, we state and proof the following theorem.
\begin{Theorem}
\label{thm:discretePIIequation}
The system~\eqref{Lax_pair}, with $T(n;z)$ of the form~\eqref{eq: Tsymmetricform} and coefficient $T_{N+1}(n)$ satisfying the symmetry condition~\eqref{eq:symmetry TN+1}, is a Lax pair for the $N$-th higher order discrete Painlev\'e~II equation and the equation is given by the expression:
\begin{equation}
\label{eq:eqdPIIhierarchy}
nx_n+\big(2x_n\Delta^{-1}(x_n-(\Delta+I)x_{n}{\rm Perm}_n)-v_n-v_n{\rm Perm}_n\big)T_{N+1,12}(n)=0,
\end{equation}
where $T_{N+1,12}(n)={L}^N(0)$ with $L$ as in~\eqref{Ln}.
\end{Theorem}

\begin{proof}
Replacing $T_{N+1,21}(n)$ with the relation~\eqref{symTi1221}, equation~\eqref{eq:DPIINTN+112et21} now reads as
\begin{equation*}
nx_n+\big(2x_n\Delta^{-1}(x_n-(\Delta+I)x_{n}{\rm Perm}_n)-v_n-v_n{\rm Perm}_n\big)T_{N+1,12}(n)=0.
\end{equation*}
Equations~\eqref{LTiTi+1} and~\eqref{TN+11221} with the relation~\eqref{symTi1221} reduce to
\begin{equation*}
T_{i+1,12}(n)=L(T_{i,12}(n))\qquad \text{and}\qquad
T_{N+1,12}(n)={L}^N(0),
\end{equation*}
which concludes the proof.
\end{proof}

The next two examples explain for $N=1,2$ how to compute explicitly equation~\eqref{eq:eqdPIIhierarchy}.

	\begin{Example}\label{ex:computationsDiscetePIIN=12}
		Using the expression defined in Theorem~\ref{thm:discretePIIequation}, we compute the first equation~\eqref{eq:hierar1 intro} and the second~\eqref{eq:hierar2 intro}.

		For $N=1$: First we compute $T_{2,12}(n)$ with the operator $L$~\eqref{Ln}:
\begin{equation*}
T_{2,12}(n)=2x_{n+1}\Delta^{-1}(0)=-\theta_1x_{n+1},
\end{equation*}
		where $-\theta_1/2$ is the integration constant.

		Replacing $T_{2,12}(n)$ in equation~\eqref{eq:eqdPIIhierarchy},
\begin{gather*}
nx_n+v_n\theta_1(x_{n+1}+x_{n-1})+2x_n\Delta^{-1}(\theta_1x_nx_{n+1}-\theta_1x_nx_{n+1})=0.
\end{gather*}
		Then
\begin{equation*}
(n+\alpha)x_n+\theta_1v_n(x_{n+1}+x_{n-1})=0.
\end{equation*}
This equation is the same as equation~\eqref{eq:hierar1 intro} if we choose the integration constant $\alpha$ to be zero.

		For $N=2$: We compute $T_{3,12}(n)$. Computations are the same for $T_{2,12}(n)$ except for the integration constant, $T_{2,12}(n)=-\theta_2x_{n+1}$.
\begin{gather*}
T_{3,12}(n)=L(T_{2,12}(n))=\big(x_nx_{n+1}^2-v_{n+1}x_{n+2}\big)\theta_2
\\ \hphantom{T_{3,12}(n)=}
{}+x_{n+1}\big(2\Delta^{-1}+I\big)(-\theta_2x_nx_{n+1}+\theta_2x_nx_{n+1})
\end{gather*}
Then $T_{3,12}(n)=\theta_2\big(x_nx_{n+1}^2-v_{n+1}x_{n+2}\big)-\theta_1x_{n+1}$.

Replacing $T_{3,12}(n)$ in equation~\eqref{eq:eqdPIIhierarchy},
\begin{equation*}
(n+\alpha)x_n+\theta_2v_n\big(v_{n+1}x_{n+2}+v_{n-1}x_{n-2}-x_n(x_{n+1}+x_{n-1})^2\big)
+\theta_1v_n(x_{n+1}+x_{n-1})=0
\end{equation*}
		which is the same equation as~\eqref{eq:hierar2 intro}.
	\end{Example}
We finally conclude the work by noticing that Theorem~\ref{thm:discretePIIequation} together with Corollary~\ref{cor:toeplitz recursion relation xn} give the proof of Theorem~\ref{thm:main intro}.

	\begin{Remark}
		In our setting, the fixed $N\geq 1$ define the order $(2N)$ of the discrete equation solved by $x_n$, the quantity related to the Toeplitz determinants $D_n$. An alternative approach could be to leave $N$ variate and consider it as a second discrete variable for $x_n$. In effect, this is done in~\cite{HisakadoWadatireal}, where the authors consider orthogonal polynomials on the real line, w.r.t.\ a weight $\rho(\lambda;N){\rm d}\lambda$ and where the dependence on an integer parameter $N$ is such that $\rho(\lambda;N+1)=\lambda \rho(\lambda;N)$. In this case the relevant quantities to consider (related to the Hankel determinants) are the coefficients of the three terms recurrence relation satisfied by these polynomials. The authors there proved that these quantities solve (up to some change of variables) the discrete-time Toda molecule equation, a coupled system of discrete equations in the two variables $n$, $N$. The result deeply relies on the quasi-periodic condition satisfied by the weight $\rho$. Back to our setting, the measure we have for our orthogonal polynomials on the unit circle is such that
\[
\d \mu(\lambda;N+1)= \e^{\sum_{j=1}^{N+1} \frac{\theta_j}{j}(\e^{{\rm i} \lambda j}+\e^{-{\rm i} \lambda j})}\frac{\d \lambda}{2\pi}= \e^{\frac{\theta_{N+1}}{N+1}(\e^{{\rm i}\lambda (N+1)}+\e^{-{\rm i}\lambda (N+1)})}\, \d \mu(\lambda;N).
\]
		This relation does not seem as promising as the one for $\rho$ for the study of the $N$-dependence, but it is another point that we could further investigate.
\end{Remark}

\appendix
\section{The continuous limit}
\label{appendix}
This appendix contains further computations for the continuous limit of the equations of the discrete Painlev\'e~II hierarchy~\eqref{eq:dPIIhierarchy intro} in the first cases $N=1,2,3$. To obtain it, we follow the scaling limit given in~\cite[Theorem 1]{betea2021multicritical} as already recalled in the introduction.

{\bf The case $\boldsymbol{N=1}$.} Notice that in this case we recover the same computation done in~\cite[Chapter 9]{Borodindiscrete}.
We consider equation~\eqref{eq:hierar1 intro} written as
\begin{equation*}
	x_{n+1}+x_{n-1}+\frac{nx_n}{\theta_1\big(1-x_n^2\big)}=0
\end{equation*}
in which the only parameter appearing is $\theta_1=\theta$. Following the scaling limit of~\cite[Theorem~1]{betea2021multicritical}, in the case $N=1$, we have
\begin{equation*}
	b=2,\qquad d=1\qquad\text{and}\qquad x_n = (-1)^n \theta^{-\frac{1}{3}}u(t) \quad\text{with}\quad t=(n-2\theta)\theta^{-\frac{1}{3}}.
\end{equation*}
Now, for $\theta\rightarrow +\infty$, we compute
\begin{align*}
	x_{n\pm 1} &\sim (-1)^{n+1}\theta^{-\frac{1}{3}} u\big(t\pm \theta^{-\frac{1}{3}}\big)\\
	&\sim (-1)^{n+1}\theta^{-\frac{1}{3}} \bigg(u(t)\pm\theta^{-\frac{1}{3}}u'(t) +\frac{\theta^{-\frac{2}{3}}}{2}u''(t)+O\big(\theta^{-1}\big)\bigg),
\end{align*}
that gives
\begin{equation*}
	x_{n+1}+x_{n-1} \sim (-1)^{n+1}2\theta^{-\frac{1}{3}}u(t)+ (-1)^{n+1}\theta^{-1}u''(t) + O\big(\theta^{-1}\big).
\end{equation*}
The other term appearing in the discrete Painlev\'e~II equation gives instead
\begin{align*}
	\frac{nx_n}{\theta_1\big(1-x_n^2\big)}\sim \big(2\theta +t\theta^{\frac{1}{3}}\big)(-1)^n \theta^{-\frac{1}{3}}u(t) \theta^{-1}\bigg(1 + \theta^{-\frac{2}{3}}u^2(t) +O\big(\theta^{-1}\big)\bigg)
\\
	\sim (-1)^n2\theta^{-\frac{1}{3}}u(t)+(-1)^n \theta^{-1}\big(tu(t) + 2u^3(t)\big) +O\big(\theta^{-1}\big).
\end{align*}
Thus equation~\eqref{eq:dPIIintro} in this scaling limit gives at the first order (coefficient of $\theta^{-1}$) the second order differential equation
\begin{equation*}
	u''(t)-tu(t)-2u^3(t) =0,
\end{equation*}
which coincides indeed with the Painlev\'e~II equation.

{\bf The case $\boldsymbol{N=2}$.} We consider equation~\eqref{eq:hierar2 intro}, with the parameters $\theta_1$, $\theta_2$ rescaled as
$\theta_1=\theta$, $\theta_2 = \frac{\theta}{4}$. It reads as
\begin{align}
	&\frac{n x_n}{\big(1-x_n^2\big)}+\theta(x_{n+1}+x_{n-1})\nonumber
\\
&{}\qquad+\frac{\theta}{4}\big(x_{n+2}\big(1-x_{n+1}^2\big)+x_{n-2}\big(1-x_{n-1}^2\big)-x_n(x_{n+1}+x_{n-1})^2\big)=0
	\label{eq:dpIIN=2scale}
\end{align}
and this time we consider the following scaling limit (case $N=2$ in~\cite[Theorem~1]{betea2021multicritical})
\begin{equation*}
	b=\frac{3}{2}, \qquad
d=4\qquad\text{and}\qquad x_n=(-1)^n\theta^{-\frac{1}{5}}4^{\frac{1}{5}} u(t)
\quad\text{with}\quad t=\bigg(n-\frac{3}{2}\theta\bigg)\theta^{-\frac{1}{5}}4^{\frac{1}{5}}.
\end{equation*}
For $\theta \rightarrow +\infty$, similar computations gives the fourth order differential equation
\begin{equation*}
	t u(t) + 6 u(t)^5 - 10 u(t)u'(t)^2 -
	10 u(t)^2 u''(t) + u''''(t)=0
\end{equation*}
which corresponds to the second equation of the Painlev\'e~II hierarchy. Detailed computations to obtain certain terms from the previous equation are given below. We begin with the expansion of the first term in equation~\eqref{eq:dpIIN=2scale}:
\begin{align*}
	\frac{nx_n}{\big(1-x_n^2\big)}&\sim \bigg(\frac{3}{2}\theta+4^{-\frac{1}{5}}\theta^{\frac{1}{5}}t\bigg)(-1)^n\theta^{-\frac{1}{5}}4^{\frac{1}{5}} u(t) \big(1 + 4^{\frac{2}{5}}\theta^{-\frac{2}{5}}u^2(t) +4^{\frac{4}{5}}\theta^{-\frac{4}{5}}u^4(t)+O\big(\theta^{-1}\big)\big)
\\
	&\sim (-1)^n\bigg(\frac{3}{2}4^{\frac{1}{5}}\theta^{\frac{4}{5}} u(t)+\frac{3}{2}4^{\frac{3}{5}}\theta^{\frac{2}{5}} u(t)^3+tu(t)+6u(t)^5+O\big(\theta^{-\frac{1}{5}}\big)\bigg).
\end{align*}
Computing expansions of $x_{n\pm 1}$, $x_{n\pm 2}$ as $\theta\to\infty$, we obtain
\begin{align*}
&x_{n\pm 1} \sim (-1)^{n+1}4^{\frac{1}{5}}\theta^{-\frac{1}{5}} u\big(t\pm 4^{\frac{1}{5}}\theta^{-\frac{1}{5}}\big)
\sim (-1)^{n+1}4^{\frac{1}{5}}\theta^{-\frac{1}{5}}
\\
&\hphantom{x_{n\pm 1} \sim}
\times \bigg(u(t)\pm4^{\frac{1}{5}}\theta^{-\frac{1}{5}}u'(t) +\frac{4^{\frac{2}{5}}\theta^{-\frac{2}{5}}}{2}u''(t)\pm\frac{4^{\frac{3}{5}}\theta^{-\frac{3}{5}}}{6}u'''(t) +\frac{4^{\frac{4}{5}}\theta^{-\frac{4}{5}}}{24}u''''(t)+O\big(\theta^{-1}\big)\bigg),
\\
&	x_{n\pm 2}\sim (-1)^{n}4^{\frac{1}{5}}\theta^{-\frac{1}{5}} u\big(t\pm 2\theta^{-\frac{1}{5}}4^{\frac{1}{5}}\big)
\sim (-1)^{n}4^{\frac{1}{5}}\theta^{-\frac{1}{5}}
\\
&\hphantom{x_{n\pm 2}\sim}
\times\bigg(u(t)\pm 4^{\frac{1}{5}}2\theta^{-\frac{1}{5}}u'(t) +4^{\frac{7}{5}}\theta^{-\frac{2}{5}}u''(t)\pm\frac{4^{\frac{8}{5}}2\theta^{-\frac{3}{5}}}{3}u'''(t) +\frac{4^{\frac{9}{5}}\theta^{-\frac{4}{5}}}{3}u''''(t)+O\big(\theta^{-1}\big)\bigg)
\end{align*}
that gives for the second term of equation~\eqref{eq:dpIIN=2scale}
\begin{equation*}
	\theta(x_{n+1}+x_{n-1}) \sim (-1)^{n+1}\bigg(4^{\frac{1}{5}}2\theta^{\frac{4}{5}}u(t)+4^{\frac{3}{5}}\theta^{\frac{2}{5}}u''(t) +\frac{1}{3}u''''(t)+O\big(\theta^{-\frac{1}{5}}\big)\bigg).
\end{equation*}
Some linear and nonlinear terms appear with the expansion of the third term of equation~\eqref{eq:dpIIN=2scale}. The linear one is
\begin{equation*}
\frac{\theta}{4}(x_{n+2}+x_{n-2}) \sim (-1)^{n}\bigg(4^{\frac{1}{5}}\theta^{\frac{4}{5}}\frac{1}{2}u(t) +4^{\frac{3}{5}}\theta^{\frac{2}{5}}u''(t)+\frac{4}{3}u''''(t)+O\big(\theta^{-\frac{1}{5}}\big)\bigg).
\end{equation*}
Nonlinear ones are
\begin{align*}
&\frac{\theta}{4}x_n(x_{n+1}+x_{n-1})^2 \sim (-1)^{n} u(t)\big(4^{\frac{3}{5}}\theta^{\frac{2}{5}}u(t)^2+4 u(t)u''(t)+O\big(\theta^{-\frac{1}{5}}\big)\big),
\\
&\frac{\theta}{4}x_{n\pm 2}x_{n\pm 1}^2\sim(-1)^n\big(4^{-\frac{2}{5}}\theta^{\frac{2}{5}}u(t)^3\pm 4^{\frac{4}{5}}\theta^{\frac{1}{5}}u(t)^2u'(t)+3u(t)^2u''(t)+5u(t)u'(t)^2\big).
\end{align*}
From these computations, we see that we recover exactly
\begin{equation*}
t u(t) + 6 u(t)^5 - 10 u(t)u'(t)^2 -10 u(t)^2 u''(t) + u''''(t)=0.
\end{equation*}

{\bf The case $\boldsymbol{N=3}$.}
We consider equation~\eqref{eq:hierar3 intro} with the parameters $\theta_1$, $\theta_2$, $\theta_3$ rescaled as
$\theta_1=\theta$, $\theta_2 = \frac{2\theta}{5}$, $\theta_3=\frac{\theta}{15}$ and rewritten as
\begin{align*}
&\frac{n x_n}{\theta\big(1-x_n^2\big)}+(x_{n+1}+x_{n-1}) +\frac{2}{5}\big(x_{n+2}\big(1-x_{n+1}^2\big)+x_{n-2}\big(1-x_{n-1}^2\big)-x_n(x_{n+1}+x_{n-1})^2\big)
\\	
&\qquad +\frac{1}{15}\big(x_n^2(x_{n+1}+x_{n-1})^3+x_{n+3}\big(1-x_{n+2}^2\big)\big(1-x_{n+1}^2\big) +x_{n-3}\big(1-x_{n-2}^2\big)\big(1-x_{n-1}^2\big)\big)
\\
&\qquad +\frac{1}{15}\bigl(-2x_n(x_{n+1}+x_{n-1})\big(x_{n+2}\big(1\!-x_{n+1}^2\big)\!+x_{n-2}\big(1\!-x_{n-1}^2\big)\big) -x_{n-1}x_{n-2}^2\big(1\!-x_{n-1}^2\big)\bigr)
\\
&\qquad+\frac{1}{15}\bigl(-x_{n+1}x_{n+2}^2\big(1-x_{n+1}^2\big)-x_{n+1}x_{n-1}(x_{n+1}+x_{n-1})\bigr)=0.
\end{align*}
Finally, we consider the following scaling limit (case $N=3$ of~\cite[Theorem~1]{betea2021multicritical})
\begin{equation*}
b=\frac{4}{3},\qquad
d=15\qquad\text{and}\qquad
x_n=(-1)^n\theta^{-\frac{1}{7}}15^{\frac{1}{7}}u(t) \quad\text{with}\quad t=\bigg(n-\frac{4}{3}\theta\bigg)\theta^{-\frac{1}{7}}15^{\frac{1}{7}}.
\end{equation*}
Again, for $\theta\rightarrow +\infty$ the asymptotic expansion of the equation above results at the first order (coefficient of $\theta^{-1}$) into the sixth order differential equation
\begin{align*}
&	tu(t)+20u(t)^7 -140u(t)^3u'(t)^2-70u(t)^4u''(t)+70u'(t)^2u''(t)+42 u(t)u''(t)^2
\\
& \qquad{}
+56u(t)u'(t)u'''(t)+14u(t)^4u''(t)-u''''''(t)=0,
\end{align*}
which corresponds to the third equation in the Painlev\'e~II hierarchy.
\begin{Remark}
	Computations for $N=2$ and $N=3$ were performed with \texttt{Maple/Mathematica}. Files are available on demand.
\end{Remark}

\subsection*{Acknowledgments} We acknowledge the support of the H2020-MSCA-RISE-2017 PROJECT No. 778010 IPaDEGAN and the International Research Project PIICQ, funded by CNRS. During the period from November 2021 to October 2022, S.T.\ was supported also by the Fonds de la Recherche Scientifique-FNRS under EOS project O013018F and based at the Institut de Recherche en Math\'ematique et Physique of UCLouvain. The authors are grateful to Mattia Cafasso for the inspiration given to work on this project and his guidance. The authors also want to thank the referees of this paper for useful comments and suggestions. S.T.\ is also grateful to Giulio Ruzza for meaningful conversations.

\pdfbookmark[1]{References}{ref}
\LastPageEnding


\begin{thebibliography}{99}
\footnotesize\itemsep=0pt

\bibitem{AV}
Adler M., van Moerbeke P., Recursion relations for unitary integrals,
 combinatorics and the {T}oeplitz lattice, \href{https://doi.org/10.1007/s00220-003-0818-4}{\textit{Comm. Math. Phys.}}
 \textbf{237} (2003), 397--440, \href{https://arxiv.org/abs/math-ph/0201063}{arXiv:math-ph/0201063}.

\bibitem{BaikRHforlastpassageperc}
Baik J., Riemann--{H}ilbert problems for last passage percolation, in Recent
 Developments in Integrable Systems and {R}iemann--{H}ilbert Problems
 ({B}irmingham, {AL}, 2000), \textit{Contemp. Math.}, Vol.~326, \href{https://doi.org/10.1090/conm/326/05753}{Amer. Math.
 Soc.}, Providence, RI, 2003, 1--21, \href{https://arxiv.org/abs/math.PR/0107079}{arXiv:math.PR/0107079}.

\bibitem{baik1999distribution}
Baik J., Deift P., Johansson K., On the distribution of the length of the
 longest increasing subsequence of random permutations, \href{https://doi.org/10.1090/S0894-0347-99-00307-0}{\textit{J.~Amer. Math.
 Soc.}} \textbf{12} (1999), 1119--1178, \href{https://arxiv.org/abs/math.CO/9810105}{arXiv:math.CO/9810105}.

\bibitem{BaikDeiftSuidan}
Baik J., Deift P., Suidan T., Combinatorics and random matrix theory,
 \textit{Grad. Stud. Math.}, Vol.~172, \href{https://doi.org/10.1090/gsm/172}{Amer. Math. Soc.}, Providence, RI, 2016.

\bibitem{betea2021multicritical}
Betea D., Bouttier J., Walsh H., Multicritical random partitions,
 \textit{S\'em. Lothar. Combin.} \textbf{85 B} (2021), 33, 12~pages,
 \href{https://arxiv.org/abs/2012.01995}{arXiv:2012.01995}.

\bibitem{Borodindiscrete}
Borodin A., Discrete gap probabilities and discrete {P}ainlev\'e equations,
 \href{https://doi.org/10.1215/S0012-7094-03-11734-2}{\textit{Duke Math.~J.}} \textbf{117} (2003), 489--542,
 \href{https://arxiv.org/abs/math-ph/0111008}{arXiv:math-ph/0111008}.

\bibitem{BO99}
Borodin A., Okounkov A., A {F}redholm determinant formula for {T}oeplitz
 determinants, \href{https://doi.org/10.1007/BF01192827}{\textit{Integral Equations Operator Theory}} \textbf{37} (2000),
 386--396, \href{https://arxiv.org/abs/math.CA/9907165}{arXiv:math.CA/9907165}.

\bibitem{CafassoClaeysGirotti}
Cafasso M., Claeys T., Girotti M., Fredholm determinant solutions of the
 {P}ainlev\'e~{II} hierarchy and gap probabilities of determinantal point
 processes, \href{https://doi.org/10.1093/imrn/rnz168}{\textit{Int. Math. Res. Not.}} \textbf{2021} (2021), 2437--2478,
 \href{https://arxiv.org/abs/1902.05595}{arXiv:1902.05595}.

\bibitem{CafassoRuzza}
Cafasso M., Ruzza G., Integrable equations associated with the
 finite-temperature deformation of the discrete {B}essel point proces,
 \href{https://doi.org/10.1112/jlms.12745}{\textit{J.~Lond. Math. Soc.}}, {t}o appear, \href{https://arxiv.org/abs/2207.01421}{arXiv:2207.01421}.

\bibitem{clarkson2006lax}
Clarkson P.A., Joshi N., Mazzocco M., The {L}ax pair for the m{K}d{V}
 hierarchy, in Th\'eories Asymptotiques et \'Equations de {P}ainlev\'e,
 \textit{S\'emin. Congr.}, Vol.~14, Soc. Math. France, Paris, 2006, 53--64.

\bibitem{joshi1999discrete}
Cresswell C., Joshi N., The discrete first, second and thirty-fourth
 {P}ainlev\'e hierarchies, \href{https://doi.org/10.1088/0305-4470/32/4/009}{\textit{J.~Phys.~A}} \textbf{32} (1999), 655--669.

\bibitem{Dattoli92}
Dattoli G., Chiccoli C., Lorenzutta S., Maino G., Richetta M., Torre A.,
 Generating functions of multivariable generalized {B}essel functions and
 {J}acobi-elliptic functions, \href{https://doi.org/10.1063/1.529959}{\textit{J.~Math. Phys.}} \textbf{33} (1992),
 25--36.

\bibitem{DeiftOPRM}
Deift P.A., Orthogonal polynomials and random matrices: a {R}iemann--{H}ilbert
 approach, \textit{Courant Lect. Notes Math.}, Vol.~3, Amer. Math. Soc.,
 Providence, RI, 1999.

\bibitem{flaschkanewell1}
Flaschka H., Newell A.C., Monodromy- and spectrum-preserving deformations.~{I},
 \href{https://doi.org/10.1007/BF01197110}{\textit{Comm. Math. Phys.}} \textbf{76} (1980), 65--116.

\bibitem{FokasItsKitaev2}
Fokas A.S., Its A.R., Kitaev A.V., Discrete {P}ainlev\'e equations and their
 appearance in quantum gravity, \href{https://doi.org/10.1007/BF02102066}{\textit{Comm. Math. Phys.}} \textbf{142}
 (1991), 313--344.

\bibitem{FW04}
Forrester P.J., Witte N.S., Bi-orthogonal polynomials on the unit circle,
 regular semi-classical weights and integrable systems, \href{https://doi.org/10.1007/s00365-005-0616-7}{\textit{Constr.
 Approx.}} \textbf{24} (2006), 201--237, \href{https://arxiv.org/abs/math.CA/0412394}{arXiv:math.CA/0412394}.

\bibitem{HastingsMcLeod}
Hastings S.P., McLeod J.B., A boundary value problem associated with the second
 {P}ainlev\'e transcendent and the {K}orteweg--de {V}ries equation,
 \href{https://doi.org/10.1007/BF00283254}{\textit{Arch. Rational Mech. Anal.}} \textbf{73} (1980), 31--51.

\bibitem{HisakadoUnitary1}
Hisakado M., Unitary matrix models and {P}ainlev\'e~{III}, \href{https://doi.org/10.1142/S0217732396002976}{\textit{Modern Phys.
 Lett.~A}} \textbf{11} (1996), 3001--3010, \href{https://arxiv.org/abs/hep-th/9609214}{arXiv:hep-th/9609214}.

\bibitem{HisakadoWadatireal}
Hisakado M., Wadati M., Matrix models of two-dimensional gravity and discrete
 {T}oda theory, \href{https://doi.org/10.1142/S0217732396001788}{\textit{Modern Phys. Lett.~A}} \textbf{11} (1996), 1797--1806,
 \href{https://arxiv.org/abs/hep-th/9605175}{arXiv:hep-th/9605175}.

\bibitem{FokasItsKitaev1}
Its A.R., Kitaev A.V., Fokas A.S., An isomonodromy approach to the theory of
 two-dimensional quantum gravity, \href{https://doi.org/10.1070/RM1990v045n06ABEH002699}{\textit{Russian Math. Surveys}} \textbf{45}
 (1990), 155--157.

\bibitem{KimuraZahabi}
Kimura T., Zahabi A., Universal edge scaling in random partitions,
 \href{https://doi.org/10.1007/s11005-021-01389-y}{\textit{Lett. Math. Phys.}} \textbf{111} (2021), 48, 16~pages,
 \href{https://arxiv.org/abs/2012.06424}{arXiv:2012.06424}.

\bibitem{LDMS18}
Le~Doussal P., Majumdar S.N., Schehr G., Multicritical edge statistics for the
 momenta of fermions in nonharmonic traps, \href{https://doi.org/10.1103/PhysRevLett.121.030603}{\textit{Phys. Rev. Lett.}}
 \textbf{121} (2018), 030603, 7~pages, \href{https://arxiv.org/abs/1802.06436}{arXiv:1802.06436}.

\bibitem{Okunkov}
Okounkov A., Infinite wedge and random partitions, \href{https://doi.org/10.1007/PL00001398}{\textit{Selecta Math.
 (N.S.)}} \textbf{7} (2001), 57--81, \href{https://arxiv.org/abs/math.RT/9907127}{arXiv:math.RT/9907127}.

\bibitem{painlev00}
Painlev\'e P., M\'emoire sur les \'equations diff\'erentielles dont
 l'int\'egrale g\'en\'erale est uniforme, \href{https://doi.org/10.24033/bsmf.633}{\textit{Bull. Soc. Math. France}}
 \textbf{28} (1900), 201--261.

\bibitem{periwalschewitz}
Periwal V., Shevitz D., Exactly solvable unitary matrix models: multicritical
 potentials and correlations, \href{https://doi.org/10.1016/0550-3213(90)90676-5}{\textit{Nuclear Phys.~B}} \textbf{344} (1990),
 731--746.

\bibitem{GrammaticosdiscrPain}
Ramani A., Grammaticos B., Hietarinta J., Discrete versions of the {P}ainlev\'e
 equations, \href{https://doi.org/10.1103/PhysRevLett.67.1829}{\textit{Phys. Rev. Lett.}} \textbf{67} (1991), 1829--1832.

\bibitem{Schensted}
Schensted C., Longest increasing and decreasing subsequences,
 \href{https://doi.org/10.4153/CJM-1961-015-3}{\textit{Canadian~J.~Math.}} \textbf{13} (1961), 179--191.

\bibitem{TW94}
Tracy C.A., Widom H., Fredholm determinants, differential equations and matrix
 models, \href{https://doi.org/10.1007/BF02101734}{\textit{Comm. Math. Phys.}} \textbf{163} (1994), 33--72,
 \href{https://arxiv.org/abs/hep-th/9306042}{arXiv:hep-th/9306042}.

\bibitem{TracyWidom}
Tracy C.A., Widom H., Level-spacing distributions and the {A}iry kernel,
 \href{https://doi.org/10.1007/BF02100489}{\textit{Comm. Math. Phys.}} \textbf{159} (1994), 151--174,
 \href{https://arxiv.org/abs/hep-th/9211141}{arXiv:hep-th/9211141}.

\end{thebibliography}
\end{document}